\pdfoutput=1
\documentclass[11pt]{article}
\usepackage{fullpage}
\usepackage{amsmath,amsthm,amsfonts}
\usepackage[margin=0.95in,nohead]{geometry}
\usepackage{xcolor,url}
\usepackage{xspace,mathtools}
\usepackage{enumerate}
\usepackage{hyperref}
\usepackage{float}
\hypersetup{
    colorlinks=true, %set true if you want colored links
    linktoc=all,     %set to all if you want both sections and subsections linked
    linkcolor=blue,  %choose some color if you want links to stand out
    citecolor=red,
}

\newcommand{\maybe}[1]{}

\newcommand{\NP}{{\sf NP}}

\newcommand{\E}{{\sf E}}

\renewcommand{\P}{{\sf P}}

\newcommand{\OWF}{{\sf OWF}}

\newcommand{\PPT}{{\sf PPT}}

\newcommand{\remove}[1]{}

\newcommand{\desc}{{\Pi}}

\newcommand{\bit}{\{0,1\}}

\newtheorem{thm}{Theorem}[section]      % A counter for Theorems etc
\newcommand{\BT}{\begin{thm}}   \newcommand{\ET}{\end{thm}}
\newtheorem{dfn}[thm]{Definition}      %
\newcommand{\BD}{\begin{dfn}}   \newcommand{\ED}{\end{dfn}}
\newtheorem{corr}[thm]{Corollary}      %
\newcommand{\BCR}{\begin{corr}} \newcommand{\ECR}{\end{corr}}
\newtheorem{Ithm}{Theorem}[section]     % A counter for Theorems in Intro
\newcommand{\BIT}{\begin{Ithm}}   \newcommand{\EIT}{\end{Ithm}}
\newtheorem{lem}{Lemma}[section]  % A counter for Lemmas etc
\newcommand{\BL}{\begin{lem}}   \newcommand{\EL}{\end{lem}}
\newtheorem{prop}[lem]{Proposition}
\newcommand{\BP}{\begin{prop}}   \newcommand{\EP}{\end{prop}}
\newtheorem{clm}{Claim}         %
\newcommand{\BCM}{\begin{clm}}   \newcommand{\ECM}{\end{clm}}
\newtheorem{sclm}{SubClaim}[clm]           %
\newcommand{\BSCM}{\begin{sclm}}   \newcommand{\ESCM}{\end{sclm}}

\newtheorem{fact}[lem]{Fact}            %
\newcommand{\BF}{\begin{fact}}   \newcommand{\EF}{\end{fact}}
\renewenvironment{proof}{\noindent{\bf Proof:~~}}{\qed}
\newcommand{\BPF}{\begin{proof}} \newcommand {\EPF}{\end{proof}}
\newtheorem{prot}{Protocol}      % A counter for Protocols
\newcommand{\BPR}{\begin{prot}}   \newcommand{\EPR}{\end{prot}}
\newenvironment{cproof}{\noindent{\bf Proof:~~}}{\hfill $\Box$}
\newcommand{\BCPF}{\begin{cproof}} \newcommand {\ECPF}{\end{cproof}}

\newtheorem{remark}[lem]{Remark}            %
\newcommand{\BR}{\begin{remark}}   \newcommand{\ER}{\end{remark}}

\newcommand{\BFT}{\begin{fact}}   \newcommand{\EFT}{\end{fact}}

\newcommand{\BDE}{\begin{description}}
\newcommand{\EDE}{\end{description}}
\newcommand{\BE}{\begin{enumerate}}
\newcommand{\EE}{\end{enumerate}}
\newcommand{\BI}{\begin{itemize}}
\newcommand{\EI}{\end{itemize}}
\newcommand{\BEQ}{\begin{eqnarray*}}
\newcommand{\EEQ}{\end{eqnarray*}}
\def\blackslug
{\hbox{\hskip 1pt\vrule width 8pt height 8pt depth 1.5pt\hskip 1pt}}
\def\qed{\quad\blackslug\lower 8.5pt\null\par}

\newcommand{\cD}{{\cal D}}
\renewcommand{\H}{{\cal H}}
\newcommand{\A}{{\cal A}}

\newcommand{\cH}{{\cal H}}

\newcommand{\cR}{{\cal R}}

\renewcommand{\S}{{\cal S}}

\newcommand{\U}{{\cal U}}
\newcommand{\cV}{{\cal V}}

\newcommand{\N}{\mathbb{N}}

\newcommand{\poly}{\mathsf{poly}}

\newcommand{\SD}{\mathsf{SD}}

\newenvironment{boxfig}[2]{\begin{figure}[#1]\fbox{\begin{minipage}{\columnwidth}
                        \vspace{0.2em}
                        \makebox[0.025\columnwidth]{}
                        \begin{minipage}{0.95\columnwidth}
            {\small{
                        #2 }}
                        \end{minipage}
                        \vspace{0.2em}
                        \end{minipage}}}{\end{figure}}

\newcommand{\bitset}{\bit}

\date{}
\begin{document}

%\title{One-way Functions Exist iff $K^t$-Complexity is
 % Hard-on-Average\\
\title{On One-way Functions and Kolmogorov Complexity\\
%Exist iff $K^t$-Complexity is
 % Hard-on-Average\\
% {\bf PRELIMINARY DRAFT --- DO NOT DISTRIBUTE}
}

\author{Yanyi Liu\\Cornell University\\ \texttt{yl2866@cornell.edu}
\and Rafael Pass\thanks{Supported in part by NSF Award SATC-1704788, NSF Award
    RI-1703846, AFOSR Award FA9550-18-1-0267, and a JP Morgan Faculty Award.
This research is based upon work supported in part by the Office of the Director of National Intelligence (ODNI), Intelligence Advanced Research Projects Activity (IARPA), via 2019-19-020700006.  The views and conclusions contained herein are those of the authors and should not be interpreted as necessarily representing the official policies, either expressed or implied, of ODNI, IARPA, or the U.S. Government.  The U.S. Government is authorized to reproduce and distribute reprints for governmental purposes notwithstanding any copyright annotation therein.}\\Cornell Tech \\ \texttt{rafael@cs.cornell.edu}}
\date{\today}
\maketitle

\begin{abstract}
  \noindent We prove that  the equivalence of two fundamental
problems in the theory of computing.
For every polynomial $t(n)\geq (1+\varepsilon)n, \varepsilon>0$, the
following are equivalent:
\BI
 \item One-way functions exists (which in turn is equivalent to the existence of secure private-key
  encryption schemes, digital signatures, pseudorandom generators, pseudorandom functions, commitment schemes, and more);
\item $t$-time bounded Kolmogorov Complexity, $K^t$, is mildly
  hard-on-average (i.e., there exists a polynomial $p(n)>0$ such that no $\PPT$ algorithm can compute $K^t$, for more than a $1-\frac{1}{p(n)}$ fraction of $n$-bit strings).
\EI
In doing so, we present the first natural, and well-studied, computational problem
characterizing the feasibility of the central private-key primitives and protocols in Cryptography.
\iffalse
\item {\bf Existence of one-way functions: } the existence of one-way functions
  (which in turn is equivalent to the existence of secure private-key
  encryption schemes, digital signatures, pseudorandom generators, pseudorandom functions, commitment schemes, and more).
  \item {\bf Mild average-case hardness of $K^{\poly}$-complexity: } the existence of polynomials
%    $t(n)>2n$, $p(n)>0$
$t,p>0$ such that no $\PPT$ algorithm can determine the $t$-time
    bounded Kolmogorov Complexity, $K^t$, for more than a $1-\frac{1}{p(n)}$
      fraction of $n$-bit strings.
      \EI
      In doing so, we present the first natural, and well-studied, computational problem
      characterizing the feasibility of the central private-key
      primitives and protocols in Cryptography.
      \fi
      %: \emph{Secure private-key
 % encryption schemes, digital
 % signatures, pseudorandom generators, pseudorandom functions and
 % commitment scheme are feasible iff $K^{\poly}$ is mildly hard-on average.}
\end{abstract}
\thispagestyle{empty}
\newpage
\setcounter{page}{1}
\section{Introduction}
We prove the equivalence of two fundamental
problems in the theory of computing: (a) the existence of one-way
functions, and (b) mild average-case hardness of the time-bounded Kolmogorov Complexity problem.

\begin{description}
  \item [Existence of One-way Functions:]
A \emph{one-way function} \cite{DH76} (OWF) is a function $f$ that can be efficiently computed
(in polynomial time), yet no probabilistic polynomial-time ($\PPT$)
algorithm can invert $f$ with inverse polynomial probability
for infinitely many input lengths $n$.
Whether one-way functions exist is unequivocally the most important
open problem in Cryptography (and arguably the most importantly open
problem in the theory of computation, see e.g., \cite{levin03}): OWFs are both necessary
\cite{ImpagliazzoL89} and sufficient for
many of the most central cryptographic primitives and protocols (e.g., pseudorandom generators
\cite{BM88,HILL99}, pseudorandom functions \cite{GGM84}, private-key encryption
\cite{GM84}, digital signatures \cite{Rompel90}, commitment
schemes \cite{Naor91}, identification protocols \cite{FS90},
coin-flipping protocols \cite{Blum82}, and more). These primitives and
protocols are often referred to as \emph{private-key
  primitives}, or ``Minicrypt'' primitives \cite{Imp95} as they
exclude the notable task of public-key encryption \cite{DH76,RSA}.
Additionally, as observed by Impagliazzo \cite{G89,Imp95}, the existence of a OWF is equivalent to the existence of polynomial-time method
for sampling hard \emph{solved} instances for an $\NP$ language (i.e.,
hard instances together with their witnesses).
%\footnote{A OWF $f$
%  directly yields the desired sampling method: pick a random string
%  $r$ and let $x = f(r)$ be the instance and $r$ the
%  witness. Conversely, to see why the existence of such a sampling
%  method implies a one-way function, consider the
% function $f$ that takes the random coins used by the sampling method and
% outputs the instance generated by it.}
%

While many candidate constructions of OWFs are
known---most notably based on factoring \cite{RSA}, the discrete logarithm
problem \cite{DH76}, or the hardness of lattice problems
\cite{Ajtai96}---the question of whether there exists some \emph{natural}
average-case hard problem that characterizes the hardness of OWFs (and
thus the feasibility of the above central cryptographic primitives) has been a long-standing open problem:\footnote{Note that Levin
  \cite{Levin85} presents an ingenious construction of a \emph{universal one-way function}---a function that
  is one-way if one-way functions exists. But his construction (which
  relies on an enumeration argument) is artificial. Levin
  \cite{levin03} takes a step towards making it less artificial by
  constructing a universal one-way function based on a
  new specially-tailored {\em Tiling Expansion problem}.}
\begin{quote}
  \emph{Does there exists some \emph{natural} average-case hard
  computational problem (i.e., both the computational problem and the
  distribution over instances is ``natural''), which characterizes the
  existence of one-way functions?}
\end{quote}
This problem is particularly pressing given recent advances in quantum
computing \cite{google} and the fact that many classic OWF
candidates (e.g., based on factoring and discrete log) can be broken by a quantum computer \cite{Shor97}.

\item [Average-case Hardness of $K^{\poly}$-Complexity:]
What makes the string $12121212121212121$ less random
than $60484850668340357492$? The notion of {\em Kolmogorov complexity}
($K$-complexity), introduced by Solo\-monoff \cite{SOLOMONOFF19641}, Kolmogorov
\cite{Kolmogorov} and Chaitin \cite{Chaitin69a}, provides an elegant method for measuring the amount
of ``randomness'' in individual strings: The $K$-complexity of a
string is the length of the shortest program (to be run on some fixed universal Turing
machine $U$) that outputs the string $x$. From a computational point
of view, however, this notion is unappealing as there is no
efficiency requirement on the program. The notion of
{\em $t(\cdot)$-time-bounded Kolmogorov Complexity ($K^t$-complexity)}
overcomes this issue: $K^t(x)$ is defined as the length
of the shortest program that outputs the string $x$ within time $t(|x|)$. As surveyed
by Trakhtenbrot \cite{T84}, the problem of efficiently determining the
$K^t$-complexity for $t(n) = \poly(n)$ predates the theory of $\NP$-completeness
and
was studied in the Soviet Union since the 60s as a candidate for a
problem that requires ``brute-force search'' (see Task 5 on page 392 in \cite{T84}). The modern
complexity-theoretic study of this problem goes back to Sipser
\cite{Sipser83}, Ko
\cite{Ko86} and Hartmanis \cite{Hartmanis}.

Intriguingly, Trakhtenbrot
also notes that a ``frequential'' version of this problem was
considered in the Soviet Union in the 60s: the problem of finding an algorithm that succeeds
for a ``high'' fraction of strings $x$---in more modern terms from the
theory of average-case complexity \cite{L86}, whether
$K^t$ can be computed by a heuristic algorithm with inverse polynomial
error, over random inputs $x$. We say that $K^t$ is
{\em mildly hard-on-average (mildly HoA)} if there exists some polynomial
$p(\cdot)>0$ such that
%{\em $\alpha(\cdot)$-hard-on-average}
%if
every $\PPT$ fails in
computing $K^t(\cdot)$ for at least a $\frac{1}{p(\cdot)}$
% $\alpha(n)$
fraction of $n$-bit
strings $x$ for all sufficiently large $n$,
%raf2: added K^{poly} notation--which i have noticed people use
and that $K^{\poly}$ is
%mildly HoA if there exists some polynomial $t(n) \geq  2n$
mildly HoA if there exists some polynomial $t(n) > 0$
such that $K^t$ is mildly HoA.
%, and that $K^t$ is
%\emph{mildly hard-on-average} if there exits some polynomial
%$p(\cdot)>0$ such that $K^t$ is $\frac{1}{p(\cdot)}$-hard-on-average.
\end{description}
Our main result shows that the existence of OWFs is equivalent to
mild average-case hardness of $K^{\poly}$. In doing so, we resolve the
above-mentionned open problem, and
present the first natural (and well-studied) computational problem, characterizing the feasibility of the central private-key primitives
in Cryptography.
\noindent
\BT \label{main} The following are equivalent:
\BI
\item One-way functions exist;
\item $K^{\poly}$ is
   mildly hard-on-average.
\EI
\ET
\noindent In other words,
      \begin{quote}
        \emph{Secure private-key encryption, digial dignatures,
          pseudorandom generators, pseudorandom functions, commitment
          schemes, etc., are possible
          iff
          $K^{\poly}$-complexity is mildly hard-on-average.}
      \end{quote}

      % raf4
      In fact, our main theorem is stronger than stated: we show that
      for \emph{every} polynomial $t(n) \geq (1+\varepsilon) n$, where
      $\varepsilon>0$ is a constant, mild average-case hardness
      of $K^{t}$ is equivalent to the existence of one-way functions.
        \iffalse
As already observed in \cite{T84}, and further explored in \cite{ABK+06}, the
bounded $K^t$-complexity problem is tightly related to \emph{circuit
minimization problem} (MCSP) \cite{MCSP}---i.e., the problem of, given a truth table of a
boolean function, determining the size of the smallest circuit that
computes the function. By a slight variant of the proof of Theorem
\ref{main}, we also show that the existence of one-way functions is
equivalent to average-case hardness of the MCSP problem.
\fi

\paragraph{On the Hardness of Approximating $K^{\poly}$-complexity}
Our connection between OWFs and $K^t$-complexity has direct implications to the theory of $K^t$-complexity.
Trakhtenbrot \cite{T84} also discusses average-case hardness
of the \emph{approximate} $K^t$-complexity problem: the problem of,
given a random $x$, outputting an ``approximation'' $y$ that is
$\beta(|x|)$-close to $K^t(x)$ (i.e., $|K^t(x) - y | \leq \beta(|x|)$).
He observes that there is a trivial heuristic approximation algorithm that
succeeds with probability approaching 1 (for large enough $n$): Given
$x$, simply
output $|x|$. In fact, this trivial algorithm produces a $(d \log
n)$-approximation with probability $\geq 1 -
\frac{1}{n^d}$ over random $n$-bits strings.\footnote{At most $2^{n-d \log n}$ out of $2^n$ strings
  have $K^t$-complexity that is smaller than $n-d \log n$.}
We note that our proof that OWFs imply mild average-case hardness of $K^{\poly}$
actually directly extends to show that $K^{\poly}$ is mildly-HoA
also to $(d \log n)$-approximate.
%$p(n) > n^d$.
%where $betaan additive term of $\log n$:
%We say that $K^t$ is
%{\em $\frac{1}{p(\cdot)}$-hard-on-average to
%  $\beta(\cdot)$-approximate}, if no $\PPT$ algorithm succeeds in
%computing a $\beta(n)$-close approximation of $K^t(x)$ for more than
 % an $1-\frac{1}{p(n)}$ fraction of $n$-bit strings $x$, for infinitely many $n$.
We thus
directly get:
%\BT If there exists a polynomial $t(n)>2n$ such that $K^t$ is
%  mildly hard-on-average, then for every constant $d$,
%  there exist a polynomial $t'(n)$ such that $K^{t'}$ is mildly hard-on-average to $(d \log n)$-approximate.
%  \ET
%raf2: updated with new notation
\BT If $K^{\poly}$ is
  mildly hard-on-average, then for every constant $d$, $K^{\poly}$ is mildly hard-on-average to $(d \log n)$-approximate.
\ET
In other words, any efficient algorithm that only slightly beats the
success probability of the ``trivial" approximation algorithm, can be
used to break OWFs.
%and \emph{exactly} compute $K^{\poly}$ with overwhelming probability.

\paragraph{Existential v.s. Constructive $K^t$ complexity}
Trakhtenbrot \cite{T84} considers also ``constructive" variant of the
$K^t$-complexity problem, where the task of the solver is to, not only
determine the $K^t$-complexity of a string $x$, but to also output a
minimal-length program $\desc$ that generates $x$. We remark that for our proof that mild
average-case hardness of $K^{\poly}$ implies OWFs, it actually suffices
to assume mild average-case hardness of the ``constructive"
$K^{\poly}$ problem, and thus we obtain an equivalence between the
``existential" and ``constructive" versions of the problem in the
average-case regime.

%raf4
\paragraph{On Decisional Time-Bounded Kolmogorov Complexity Problems}
We finally note that our results also show an equivalence between one-way
functions and mild average-case hardness of a \emph{decisional}
$K^{\poly}$ problem: Let $\mathsf{MINK}^{t}[s]$ denote the set of strings $x$ such
that $K^{t(|x|)}(x) \leq s(|x|)$. Our proof directly shows that there
exists some constant $c$ such that for every constant $\varepsilon>0$,
every $t(n) \geq (1+\varepsilon) n$, and letting $s(n) = n - c\log n$, mild
average-case hardness of the language $\mathsf{MINK}^{t}[s]$ (with respect to
the uniform distribution over instances) is equivalent the existence
of one-way functions.

\subsection{Related Work}
We refer the reader to Goldreich's textbook \cite{Gol01} for more
context and applications of OWFs (and complexity-based cryptography in
general); we highly recommend Barak's survey on candidate
constructions of one-way functions \cite{Barak17}.
We refer the reader to the textbook of Li and
Vitanyi \cite{mingvitany} for more context and applications of
Kolmogorov complexity; we highly recommend Allender's
surveys on the history, and recent applications, of notions of
time-bounded Kolmogorov complexity \cite{Allender20,Allender20b,Allender17}.

\paragraph{On Connections between $K^{\poly}$-complexity and
  OWFs}
We note that some (partial) connections between $K^t$-complexity and
OWFs already existed in the literature:
\BI
\item Results by Kabanets and Cai \cite{MCSP} and
Allender et al \cite{ABK+06} show that the existence of OWFs implies
that $K^{\poly}$ must be \emph{worst-case} hard to compute; their
results will be the starting point for our result that OWFs also imply
\emph{average-case hardness} of $K^{\poly}$.
\item Allender and Das \cite{AllenderD17} show that every problem in
  ${\bf SZK}$ (the class of promise problems having statistical
  zero-knowledge proofs \cite{GMR89}) can be solved in probabilistic polynomial-time
  using a $K^{\poly}$-complexity oracle. Furthermore, Ostrovsky and
  Wigderson \cite{Ostrovsky91, OW93} show that if ${\bf SZK}$ contains
  a problem that is hard-on-average,
  then OWFs exist. In contrast, we show the existence of
  OWFs assuming only that $K^{\poly}$ is hard-on-average.
  \item A very recent elegant work by Santhanam \cite{Santhanam} is also explicitly motivated
    by the above-mentionned open problem, and presents an intruiging connection
    between one-way functions and error-less average-case hardness of the \emph{circuit
minimization problem} (MCSP) \cite{MCSP}---i.e., the problem of, given a truth table of a
boolean function, determining the size of the smallest circuit that
computes the function; the MCSP problem is closely related to the time-bounded
Kolmogorov complexity problem \cite{T84,ABK+06}. Santhanam proves
equivalence between OWFs and errorless average-case hardness of MCSP under a new (and somewhat
complicated) conjecture that he introduces. We emphasize that, in contrast, our
equivalence is unconditional.
  \EI

  \paragraph{On Worst-case to Average-case Reductions for
    $K^{\poly}$-complexity}
  We highlight a very elegant recent result by Hirahara
  \cite{Hirahara18} that presents a worst-case (approximation) to
  average-case reduction
  for $K^{\poly}$-complexity. Unfortunately, his result only gives average-case hardness w.r.t. \emph{errorless
    heuristics}---namely, heuristics that always provide either the
  correct answer or output $\bot$ (and additionally only output $\bot$
  with small probability). For our construction of a OWF, however, we require
  average-case hardness of $K^t$ also with respect to heuristics that may
  err (with small probability). Santhanam \cite{Santhanam},
  independently, obtains a similar result for a related problem.

    Hirahara notes that it is an open problem to obtain a
  worst-case to average-case reduction for $K^{\poly}$ w.r.t. heuristics that may
  err.
%raf3: added
  Let us emphasize that average-case hardness w.r.t. errorless heuristics
  is a much weaker property that just ``plain'' average-case hardness
  (with respect to heuristics that may err): Consider a random 3SAT
  formula on $n$ variables with $1000n$ clauses. It is well-known that,
  with high probability, the formula is not be satisfiable.
  Thus, there is a trivial heuristic algorithm for solving 3SAT on
  such random instances by simply outputting ``No''.
  Yet, the question of whether there exists an
  efficient \emph{errorless} heuristic for this problem is still open, and
  non-existence of such an algorithm is implied by Feige's
  Random 3SAT conjecture \cite{Feige02}.

\paragraph{On Universal Extrapolation}
Impagliazzo and Levin~\cite{ImpagliazzoL90} consider a problem of
\emph{universal extrapolation}: Roughly speaking,
extrapolation with respect to some polynomial-time Turing machine $M$ requires, given some prefix string
$x_{pre}$, sampling a random continuation $x_{post}$ such that $M$ (on input a
random tape) generates $x_{pre}||x_{post}$. Universal extrapolation is said to be possible if
\emph{all} polynomial-time Turing machines can be extrapolated.
Impagliazzo and Levin demonstrate the equivalence of one-way
functions and the infeasibility of universal extrapolation.

As suggested by an anonymous FOCS reviewer, universal extrapolation
seems related to time-bounded Kolmogorov complexity:
Extrapolation with respect to a \emph{universal}
Turing machine should, intuitively, be equivalent to \emph{approximating} $K^{\poly}$ (for random string $x$) by counting the
number of possible continuations $x_{post}$ to a prefix $x_{pre}$ of
$x$: Strings with small $K^{\poly}$-complexity should have many
possible continuation, while strings with large $K^{\poly}$-complexity
should have few.

While this method may perhaps be used to obtain an
alternative proof of one direction (existence of one-way function from
hardness of $K^{\poly}$) of our main theorem, as far as we
can tell, the actual proof is non-trivial and would result in a
significantly weaker conclusion than what we obtain: It would only show that average-case hardness of
\emph{approximating} $K^{\poly}$ implies infeasibility of universal
extrapolation and thus one-way functions, whereas we show that even
average-case hardness of \emph{exactly} computing $K^{\poly}$ implies
the existence of one-way functions.

For the converse direction, the infeasibility of universal
extrapolation only means that there exists \emph{some} polynomial-time Turing
machine $M$ that is hard to extrapolate, and this $M$ is not
necessarily a universal Turing machine. It is not \emph{a-priori}
clear whether infeasibility of extrapolation w.r.t. some $M$ implies
infeasibility of extrapolation w.r.t. a \emph{universal} Turing
machine.

A direct corollary of our main theorem is a formal connection between universal
extrapolation and average-case hardness of $K^{\poly}$:
Infeasibility of universal extrapolation is
equivalent to mild average-case hardness of $K^{\poly}$ (since by
\cite{ImpagliazzoL90}, infeasibility of universal extrapolation is
equivalent to the existence of one-way functions).

\iffalse

of some presampling

can be used to

universal extrapolation is possible, then one-way functions cannot exist.time-bounded

show that the (time-bounded) universal extrapolation is possible if and only if one way functions do not exist. As far as we know, the ability of extrapolating some universal Turing machine (over random input) implies a $c$-approximation with probability $\ge 1 - \frac{1}{2^c}$ for $K^{\poly}$-complexity. By setting the parameter properly, one may obtain that if $K^{\poly}$ is mildly-HoA to approximate, then one way functions exist, which is a weaker statement than our result (for one direction). However, for the other direction, if one way functions exist, they only show that extrapolating a pseudorandom function (with a random seed) is hard-on-average, and we show that it's also hard-on-average to compute $K^{\poly}$ (and thus, extrapolate any universal Turing machine).
\fi

\subsection{Proof outline}
We provide a brief outline for the proof of Theorem \ref{main}.

\paragraph{OWFs from Avg-case $K^{\poly}$-Hardness}
We show that if $K^t$ is mildly average-case hard for some polynomial
$t(n)>0$,
%\footnote{We require $t(n)>2n$ to perform some minimal processing to generate an $n$-bit long string.},
then a
weak one-way function exists\footnote{Recall that an efficiently
  computable function $f$ is a weak OWF if there exists some
  polynomial $q>0$ such that $f$ cannot be efficiently inverted with
  probability better than $1-\frac{1}{q(n)}$ for sufficiently large $n$.}; the existence of (strong) one-way
functions then follows by Yao's hardness amplification theorem
\cite{Yao82}. Let $c$ be a constant such that every string $x$ can be
output by a program of length $|x| + c$ (running on the fixed Universal Turing
machine $U$). Consider the function $f(\ell || \desc')$, where $\ell$
is a bitstring of
length $\log (n+c)$ and $\desc'$ is a bitstring of length $n+c$, that lets $\desc$ be the first
$\ell$ bits of $\desc'$, and outputs $\ell||y$ where $y$ is the output
generated by running the program $\desc$\footnote{Formally, the
  program/description $\Pi$ is an encoding of a pair $(M,w)$ where $M$ is a
  Turing machine and $w$ is some input, and we evaluate $M(w)$ on the
  Universal Turing machine $U$.} for
$t(n)$ steps.\footnote{We remark that although our construction of the function $f$ is
somewhat reminiscent of Levin's construction of a universal OWF, the actual function (and even more so the analysis) is
actually quite different. Levin's function $\hat{f}$, roughly speaking, parses
the input into a Turing machine $M$ of length $\log n$ and an input $x$ of length
$n$, and next outputs $M(x)$. As he argues, if a OWF $f'$ exists, then with
probability $\frac{1}{n}$, $\hat{f}$ will compute output $f'(x)$ for a
randomly selected $x$, and is thus hard to invert. In contrast, in our candidate OWF construction, the key
idea is to \emph{vary the length} of a ``fully specified'' program $\desc$
(including an input).}

We aim to show that if $f$ can be inverted with high
probability---significantly higher than $1-1/n$---then $K^t$-complexity
of \emph{random strings} $z \in \{0,1\}^n$ can be computed with high
probability. Our heuristic $\H$, given a string $z$, simply tries to
invert $f$ on $\ell||z$ for all $\ell \in [n+c]$, and outputs the
smallest $\ell$ for which inversion succeeds.\footnote{Or, in case, we
  also want to break the ``constructive'' $K^{\poly}$ problem, we also output
  the $\ell$-bit truncation of the program $\desc'$ output by the inverter.} First, note that since
every length $\ell \in [n+c]$
%Yanyi: it was [n], now I changed it to [n+c]
is selected with probability $1/(n+c)$, the
inverter must still succeed with high probability even if we condition the output of the
one-way function on any particular length $\ell$ (as we assume that
the one-way function inverter fails with probability significantly
smaller than $\frac{1}{n}$). This,
however, does not suffice to prove that the heuristic works with high
probability, as the string $y$ output by the one-way function is not
uniformly distributed (whereas we need to compute the $K^t$-complexity
for uniformly chosen strings). But, we show using a simple counting
argument that $y$ is not too ``far'' from uniform in relative distance. The key idea is
that for every string $z$ with $K^t$-complexity $w$, there exists some
program $\desc_z$ of length $w$ that outputs it; furthermore, by our
assumption on $c$, $w \leq n+c$. We thus have that $f(\U_{n+c+\log(n+c)})$ will output $w||z$ with
probability at least $\frac{1}{n+c} \cdot 2^{-w} \ge \frac{1}{n+c}
\cdot 2^{-(n+c)} = \frac{2^{-n}}{O(n)} $ (we need to pick the
right length, and next the right program). So, if the
heuristic fails with probability $\delta$, then the one-way function
inverter must fail with probability at least $\frac{\delta}{O(n)}$,
which leads to the conclusion that $\delta$ must be small (as we assumed the inverter
fails with probability significantly smaller than $\frac{1}{n}$).

\paragraph{Avg-case $K^{\poly}$-Hardness from EP-PRGs}
To show the converse direction, our starting point is the earlier
result by Kabanets and Cai \cite{MCSP} and Allender et al \cite{ABK+06} which shows that the
existence of OWFs implies that $K^t$-complexity, for every sufficiently
large polynomial $t(\cdot)$, must be \emph{worst-case}
hard to compute. In more detail, they show that if $K^t$-complexity can
be computed in polynomial-time for \emph{every} input $x$, then
pseudo-random generators (PRGs) cannot exist (and PRGs are implied by
OWF by \cite{HILL99}). This follows from the
observations that (1) random strings have high $K^t$-complexity with
overwhelming probability, and (2) outputs of a PRG always have
%raf2: added more details about t
small $K^t$-complexity as long as $t(n)$ is sufficiently greater than the running
time of the PRG (as the seed plus the constant-sized description
of the PRG suffice to compute the output). Thus,
using an algorithm that computes $K^t$, we can easily distinguish
outputs of the PRG from random strings---simply output 1 if the
$K^t$-complexity is high, and 0 otherwise. This method, however,
relies on the algorithm working for \emph{every} input. If we only have
access to a heuristic $\H$ for $K^t$, we have no guarantees that $\H$
will output a correct value when we feed it a pseudorandom string, as
those strings are \emph{sparse} in the universe of all
%raf3: added footnote
strings.\footnote{We note that, although it was not explictly pointed
  out, their argument actually also extends to show that $K^t$ does
  not have an \emph{errorless} heuristic assuming the existence of
  PRGs. The point is that even on
  outputs of the PRG, an errorless heuristic must output either a
  small value or $\bot$ (and perhaps always just output $\bot$). But
  for random strings, the heuristic can only output $\bot$ with small
  probability. Dealing with heuristics that may err will be
  more complicated.}

To overcome this issue, we introduce the concept of an
\emph{entropy-preserving PRG (EP-PRG)}. This is a PRG that expands the seed by
$O(\log n)$ bits, while ensuring that the output of the PRG loses at
most $O(\log n)$ bits of \emph{Shannon entropy}---it will be important
for the sequel that we rely on Shannon entropy as opposed to min-entropy.
In essence, the PRG preserves (up to an additive term of $O(\log n)$)
the entropy in the seed $s$. We next show that any good heuristic $\H$
for $K^t$ can break such an EP-PRG. The key point is that since the
output of the PRG is entropy preserving, by an averaging argument,
there exists a $1/n$ fraction of ``good'' seeds $S$ such that, conditioned on
the seed belonging to $S$, the output
of the PRG on input seeds of length $n$ has \emph{min-entropy} $n-O(\log n)$. This means that the probability that $\H$ fails
to compute $K^t$ on output of the PRG, conditioned on picking a ``good''
seed, can increase at most by a
factor $poly(n)$. We conclude that $\H$
can be used to determine (with sufficiently high probability) the $K^t$-complexity for both random strings and for outputs of the PRG.

\paragraph{EP-PRGs from Regular OWFs}
%Constructing an EP-PRG from OWFs, however, is a bit tricky.
We start by noting that the standard Blum-Micali-Goldreich-Levin \cite{BM84,GL89}
PRG construction from one-way \emph{permutations} is entropy
preserving. To see this, recall the construction:
% $$G_f(s,h_{GL}) = f(s) || h_{GL}(s)$$
% yanyi3: h_GL should appear in the output
$$G_f(s,h_{GL}) = f(s) || h_{GL} || h_{GL}(s)$$
where $f$ is a one-way
permutation and $h_{GL}$ is a hardcore function for $f$---by
\cite{GL89}, we can select a random hardcore function $h_{GL}$ that output $O(\log n)$
bits.
Since $f$ is a permutation, the output of the
PRG fully determines the input and thus there is actually no entropy
loss. We next show that the PRG construction of \cite{GKL93,HILL99,Gol01,Yu} from \emph{regular} OWFs also
is an EP-PRG. We refer to a function $f$ as being $r$-regular if for every
% $x \in \{0,1\}^*$, $f(x)$ has between $2^{r(n)-1}$ and $2^{r(n)}$ many preimages.
% yanyi: minors
$x \in \{0,1\}^*$, $f(x)$ has between $2^{r(|x|)-1}$ and $2^{r(|x|)}$ many preimages.
Roughly speaking, the construction
applies pairwise independent hash functions (that act as strong
extractors) $h_1,h_2$ to both the input and output of
the OWF (parametrized to match the regularity $r$) to
``squeeze'' out randomness from both the input and the output, and
finally also applies a hardcore function that outputs $O(\log n)$
bits:
\begin{equation}
\label{GKL.eq}
  G_f^r(s||h_1||h_2||h_{GL})=h_{GL}|| h_1||h_2||[h_1(s)]_{r-O(\log n)} || [h_2(f(s))]_{n-r-O(\log n)} ||
  h_{GL}(s),
  \end{equation}
where $[a]_j$ means $a$ truncated to $j$ bits.
As already shown in \cite{Gol01} (see also \cite{Yu}), the output of the
function excluding the hardcore bits is actually
% Yanyi: I changed it to 1/poly(n) since the 'entropy loss' in LHL here is O(\log n)
$1/\poly(n)$
% $1/n^2$
-close to
uniform in statistical distance (this follows directly from the
Leftover Hash Lemma \cite{HILL99}), and this implies (using an
averaging argument) that the Shannon entropy of the output is at least $n-O(\log
n)$, thus the construction is an EP-PRG.
We finally note that this construction remains both secure and
entropy preserving, even if the
input domain of the function $f$ is not $\{0,1\}^n$, but rather
\emph{any} set $S$ of size $2^n/n$; this will be useful to us shortly.

\paragraph{Cond EP-PRGs from Any OWFs}
Unfortunately, constructions of PRGs from OWFs
\cite{HILL99,Hol06,HHR06,HRV10} are not entropy preserving as far as
we can tell. We, however, remark that to prove that $K^t$ is
%yanyi2: added 'mild' here
mildly HoA, we do not
actually need a ``full-fledged'' EP-PRG: Rather, it suffices to have
what we refer to as a \emph{conditionally-secure} EP-PRG $G$: a conditionally-secure EP-PRG (cond EP-PRG) is an
efficiently computable function $G$ having the property that there
exists some event $E$ such that:
\begin{enumerate}
%yanyi2: notations of n, m, N are a little confusing
% just making them consistent with what we have below
% n: input length of f, N: input length of G, m: length of hash func
% M: output length of G
% raf3: switching to n' and m; big letter seem like r.v.
   \item $G(\U_{n'}\mid E)$ has Shannon entropy $n'-O(\log n')$;
   \item $G(\U_{n'} \mid E)$ is indistinguishable from $\U_{m}$ for some
     $m \geq n' + O(\log n')$.
%   \item $G(\U_N\mid E)$ has Shannon entropy $N-O(\log N)$;
%   \item $G(\U_N \mid E)$ is indistinguishable from $\U_M$ for some
%     $M \geq N + O(\log N)$.
  \end{enumerate}
  In other words, there exists some event $E$ such that conditionned
  on the event $E$, $G$ behaves likes an EP-PRG.
    We next show how to adapt the above construction to yield a cond EP-PRG
    from any OWF $f$. Consider $G(i||s||h_1,h_2,h_{GL}) = G_f^i(s,
    h_1,h_2,h_{GL})$ where $|s|=n$, $|i|=\log n$, and $G_f^i$ is the
    PRG construction defined in equation \ref{GKL.eq}.
    % yanyi: adding more details about i
%    and $i \in [n]$.
We remark that for any function $f$, there exists some regularity
$i^*$ such that at least a fraction $1/n$ of inputs $x$ have
regularity $i^*$. Let $S_{i^*}$ denote the set of these
$x$'s. Clearly, $|S_{i^*}| \geq 2^n/n$; thus, by the above argument,
%f$ is one-way even if we restrict to $x\in S_{i^*}$
%(as this set is dense). Thus,
%$\{x \leftarrow U_n | x \in S_{i^*} :
%G^{i^*}(x) \}$
$G_f^{i^*}(\U_{n'} \mid  S_{i^*})$ is both
pseudorandom and has entropy
%yanyi2: changing n to N
%raf2: n'
$n'-O(\log n'$).
%$N-O(\log N$).
Finally, consider the event $E$ that $i = i^*$ and $s \in S_{i^*}$. By
definition,
% yanyi: \U_{m} means those hash function descriptions
$G(\U_{\log n} ||\U_{n} ||\U_{m} \mid E )$
% $G(\U_{\log n} ||\U_{n} \mid E )$
is identically distributed to
$G_f^{i^*}(\U_{n'} \mid S_{i^*})$, and thus $G$ is a cond EP-PRG from any OWF.
For clarity, let us provide the full expanded description of the cond
EP-PRG $G$:
$$G(i||s||h_1||h_2||h_{GL})  = h_{GL}|| h_1||h_2||[h_1(s)]_{i-O(\log n)} || [h_2(f(s))]_{n-i-O(\log
  n)} || h_{GL}(s)$$
Note that this $G$ is \emph{not} a PRG: if the input $i\neq i^*$
(which happens with probability $1-\frac{1}{n}$), the output of $G$ may
not be pseudorandom! But, recall that the notion of a \emph{cond}
EP-PRG only requires the output of $G$ to be pseudorandom
\emph{conditioned} on some event $E$ (while also being entropy
preserving conditioned on the same event $E$).

Finally, the above outline only shows that $K^t$ is mildly HoA if
$t(\cdot)$ is larger than running time of the cond EP-PRG that we
constructed; that is, so far, we have only shown that OWFs imply that
$K^t$ is mildly HoA for some polynomial $t$. To prove that this holds
for every $t(n)\geq (1+\varepsilon) n$, $\varepsilon>0$, we remark that using
a padding trick, we can also construct a cond EP-PRG that can be
computed in time $n + O(n^{\alpha})$, where $\alpha<1$---we refer to
  this as a \emph{rate-1 efficient PRG}. Using such a rate-1 efficient
  cond EP-PRG, we can show that $K^t$ is mildly HoA for every
  $t(n)\geq (1+\varepsilon) n$, $\varepsilon>0$.

\section{Preliminaries}
We assume familiarity with basic concepts such as Turing machines,
polynomial-time algorithms and
probabilistic polynomial-time algorithms ($\PPT$).
%non-uniform
%polynomial-time and non-uniform $\PPT$ algorithms.
A function $\mu$ is said to be \emph{negligible} if for every
polynomial $p(\cdot)$ there exists some $n_0$ such that for all $n >
n_0$, $\mu(n) \leq \frac{1}{p(n)}$.
A {\em probability ensemble} is a sequence of random variables
$A=\{A_n\}_{n\in \N}$. We let $\U_n$ the uniform distribution over $\{0,1\}^n$.

\subsection{One-way Functions}
We recall the definition of one-way functions \cite{DH76}. Roughly speaking, a
function $f$ is one-way if it is polynomial-time computable, but hard to
invert for $\PPT$ attackers. %The standard (cryptographic) definition
%of a one-way function (see e.g., \cite{Gol01})
%requires every $\PPT$ attacker to fail (with high probability) on all
%sufficiently large input lengths.
%We will
%also consider a weaker notion of an \emph{infinitely-often} one-way
%function \cite{OstrovskyW93} which only requires the $\PPT$ attacker to fail for
%infinitely many inputs length (in other words, there is no $\PPT$
%attacker that succeeds on all sufficiently large input lengths, analogously to
%complexity-theoretic notions of hardness).

\BD\label{def:owf} Let $f: \bitset^* \rightarrow \bitset^*$ be a polynomial-time
computable function. $f$ is said to be a \emph{one-way function (OWF)} if for every $\PPT$
algorithm $\A$, there exists a negligible function $\mu$ such that for
all $n \in \N$,
	$$ \Pr[x \leftarrow \bitset^n; y = f(x) : A(1^n,y) \in f^{-1}(f(x)) ] \leq \mu(n) $$
        %      $f$ is said to be an \emph{infinitely-often one-way function (ioOWF)} if the above
%condition holds for infinitely many $n\in\NN$ (as opposed to all).
        \ED

        We may also consider a weaker notion of a \emph{weak one-way
        function} \cite{Yao82}, where we only require all $\PPT$ attackers to fail
        with probability noticeably bounded away from 1:

        \BD\label{def:owf} Let $f: \bitset^* \rightarrow \bitset^*$ be a polynomial-time
computable function. $f$ is said to be a \emph{$\alpha$-weak one-way
  function ($\alpha$-weak OWF)} if for every $\PPT$
algorithm $\A$, for all sufficiently large $n \in N$,
	$$ \Pr[x \leftarrow \bitset^n; y = f(x) : A(1^n,y) \in f^{-1}(f(x)) ] < 1-\alpha(n) $$
We say that $f$ is simply a \emph{weak one-way function (weak OWF)} if
there exists some polynomial $q>0$ such that $f$ is a
$\frac{1}{q(\cdot)}$-weak OWF.
\ED

Yao's hardness amplification theorem \cite{Yao82} shows that any weak
OWF can be turned into a (strong) OWF.
\BT [\cite{Yao82}]
Assume there exists a weak one-way function. Then there exists a
one-way function.
\ET

\subsection{Time-bounded Kolmogorov Complexity}
Let $U$ be some fixed Universal Turing machine that can emulate any Turing machine $M$ with polynomial overhead. Given a
description $\desc \in \{0,1\}^*$ which encodes a pair $(M,w)$ where
$M$ is a (single-tape) Turing
machine and $w \in \{0,1\}^*$ is an input, let $U(\desc, 1^t)$ denote the
output of $M(w)$ when emulated on $U$ for $t$ steps. Note that (by
assumption that $U$ only has polynomial overhead)
$U(\desc,1^t)$ can be computed in time $\poly(d,t)$.

The \emph{$t$-time bounded Kolmogorov Complexity, $K^t(x)$, of a string
$x$} \cite{Kolmogorov,Sipser83,T84,Ko86} is defined as the length of the shortest description $\desc$
such that $U(\desc,1^t) = x$:
$$K^t(x) = \min_{\desc \in \{0,1\}^*}\{|\desc|:U(\desc, 1^{t(|x|)})=x\}.$$
A central fact about $K^t$-complexity is that the length of a
string $x$ essentially (up to an additive constant) bounds the
$K^t$-complexity of the string for every $t(n)>0$
\cite{SOLOMONOFF19641,Kolmogorov, Chaitin69a}
(see e.g., \cite{Sipser96} for simple treatment). This follows by considering $\desc =
(M,x)$ where $M$ is a constant-length Turing machine that directly
halts; consequently, $M$ simply outputs its input and thus $M(x) = x$.
\begin{fact} \label{KC.fact}
There exists a constant $c$ such that for every function
$t(n)>0$ and every $x \in \{0,1\}^*$ it holds that $K^t(x) \leq |x|+c$.
\end{fact}

\iffalse
%raf: all this is correct, but does not suffice for us, we need
%additive overhead.
\subsection{Circuit Minimization}
Given the truth-table of a boolean function $f: \{0,1\}^n \rightarrow
\{0,1\}$, let $MCS(f)$ (minimum circuit size) denote the size of the
smallest boolean circuit with AND, OR and NOT gates that computes
$f$, where $f$ is provided as a truth-table.
(The function basis is not important to our results, as long as
the function basis is complete).

Lupanov's theorem \cite{Lupanov} gives a tight upper bound on the circuit size of any
boolean function: For any boolean function $f: \{0,1\}^n \rightarrow
\{0,1\}$, $MCS(f) \leq O(2^n/n)$. As a boolean circuit of size $S$ can
be described by a string of length $O(S \log |S|)$, we directly get
the following fact.:

\begin{fact} \label{MCS.fact} There exists a constant $c$ such that every boolean function $f:
\{0,1\}^n \rightarrow \{0,1\}$ can be described by a string of length
$c 2^n$.
\end{fact}
\fi
\subsection{Average-case Hard Functions}
We turn to defining what it means for a function to be average-case hard
(for $\PPT$ algorithms).
\BD
We say that a function $f:\bitset^* \rightarrow \bitset^*$ is \emph{$\alpha(\cdot)$
hard-on-average ($\alpha$-HoA)} if for all $\PPT$ heuristic $\H$, for all
sufficiently large $n \in N$,
$$\Pr[x \leftarrow \{0,1\}^n : \H(x) = f(x)] < 1 - \alpha(|n|)$$
\ED
In other words, there does not exist a $\PPT$ ``heuristic'' $\H$
that computes $f$ with probability $1-\alpha(n)$ for infinitely many
$n \in N$.
We also consider what it means for a function to be average-case hard
to \emph{approximate}.
\BD
We say that a function $f:\bitset^* \rightarrow \bitset^*$ is \emph{$\alpha$
hard-on-average ($\alpha$-HoA) to $\beta(\cdot)$-approximate} if for all $\PPT$ heuristic $\H$, for all
sufficiently large $n \in N$,
$$\Pr[x \leftarrow \{0,1\}^n : |\H(x) - f(x)| \leq \beta(|x|) ] < 1 - \alpha(|n|)$$
\ED
In other words, there does not exists a $\PPT$ heuristic $\H$
that approximates $f$ within a $\beta(\cdot)$ additive term, with
probability $1-\alpha(n)$ for infinitely many $n \in N$.

Finally, we refer to a function $f$ as being \emph{mildly} HoA (resp
HoA to approximate) if there exists a polynomial
$p(\cdot)>0$ such that $f$ is $\frac{1}{p(\cdot)}$-HoA
  (resp. HoA to approximate).
\subsection{Computational Indistinguishability}
We recall the definition of (computational) indistinguishability \cite{GM84}.

\BD
\label{ind.def}
Two ensembles
$\{A_{n}\}_{n \in \N}$ and
$\{B_{n}\}_{n \in \N}$
are said to be {\em $\mu(\cdot)$-indistinguishable},
if for every probabilistic machine $D$ (the ``distinguisher'')
whose running time is polynomial
in the length of its first input,
there exists some $n_0 \in \N$
%there exists a negligible function
%$\mu(\cdot)$
so that for every $n\geq n_0$:
\begin{eqnarray*}
\left| \Pr[D(1^n,A_{n})=1] - \Pr[D(1^n,B_{n})=1] \right| < \mu(n)
\end{eqnarray*}
We say that are $\{A_{n}\}_{n \in \N}$ and
$\{B_{n}\}_{n \in \N}$ simply \emph{indistinguishable}
if they are $\frac{1}{p(\cdot)}$-indistinguishable for every
polynomial $p(\cdot)$.
\ED

\subsection{Statistical Distance and Entropy}
For any two random variables $X$ and $Y$ defined over some set $\cV$, we let %$||X-Y||$
%$\mathsf{SD}(X,Y) = \max_{T \subseteq U}|\Pr[X \in T]-\Pr[Y\in T]|$
%denote the
$\mathsf{SD}(X,Y) = \frac{1}{2} \sum_{v \in \cV} |\Pr[X = v] -\Pr[Y = v]|$ denote the
\emph{statistical distance} between $X$ and $Y$.
For a random variable $X$, let
%$H_{\infty}(X) = \min_{x \in Supp(X)}
%\log \frac{1}{\Pr[X=x]}$ denote the \emph{min entropy} of $X$.
%, defined as $||X-Y|| = \max_{T \subset U}|\Pr[X \in T]-\Pr[Y\in T]|$.
$H(X) = \E [ \log \frac{1}{\Pr[X=x]} ]$ denote the (Shannon) entropy
of $X$, and let $H_{\infty}(X) = \min_{x \in Supp(X)}
\log \frac{1}{\Pr[X=x]}$ denote the \emph{min-entropy} of $X$.

We next demonstrate a simple lemma showing that any distribution that is
statistically close to random, has very high Shannon entropy.
\BL \label{lemma:SDtoH}
    For every $n\geq 4$, the following holds. Let
    $X$ be a random variable over $\{0, 1\}^n$ such that $\SD(X,\U_n) \leq \frac{1}{n^2}$. Then  $H(X_n) \ge n - 2.$
\EL
\begin{proof}
    Let $S = \{x \in \{0, 1\}^n : \Pr[X = x] \leq 2^{-(n - 1)}\}$.
    Note that for every $x \notin S$, $x$ will contribute at
    least
$$\frac{1}{2} \left( \Pr[X = x]  - \Pr[U_n = x] \right) \geq \frac{1}{2}
\left( \Pr[X =
x]  - \frac{\Pr[X =
  x]}{2} \right) = \frac{\Pr[X =x]}{4}$$
to $SD(X,\U_n)$. Thus,
    $$\Pr[X \notin S] \leq 4 \cdot \frac{1}{n^2}.$$
    Since for every $x \in S$, $\log \frac{1}{\Pr[X=x]} \geq n-1$ and the probability that $X \in S$ is at least $1 - 4/n^2$, it follows that
    $$H(X) \ge \Pr[X \in S](n-1) \ge (1 - \frac{4}{n^2}) (n-1) \ge n -
    \frac{4}{n} -1 \geq  n-2.$$
\end{proof}

\iffalse
\subsection{Pairwise Independent Hashing}
\BD [Efficient Family of Pairwise Independent Hash Functions] Let $\cal H$ be a family of functions where each function $h \in {\cal H}$ goes from $\bitset^m$ to $\bitset^n$. We say that
$\cal H$ is a an efficient family of pairwise independent hash functions if (i) the functions $h \in {\cal H}$ can be described with a polynomial (in $n$) number of bits; (ii) there is
a polynomial (in $n$) time algorithm to sample a random $h \in {\cal H}$; (iii) for all $x \neq x' \in \bitset^m$ and for all $y,y' \in \bitset^n$
$$ \Pr[h \leftarrow \cH:h(x)=y \; \mbox{ and } \; h(x')=y'] = 2^{-2n} $$
\ED
\fi
\section{The Main Theorem}

\BT \label{main.thm} The following are equivalent:
\BI
\item [(a)] The existence of one-way functions.
\item [(b)] The existence of a polynomial $t(n)>0$ such that $K^t$ is
  mildly hard-on-average.
%raf2: updated bullet c to become stronger.
\item [(c)] For all constants $d>0, \varepsilon>0$, and every polynomial $t(n) \geq (1+\varepsilon) n$, $K^t$ is mildly hard-on-average to
  $(d \log n)$-approximate.
\EI
\ET

We prove Theorem \ref{main.thm} by showing that (b)
implies (a) (in Section \ref{OWF.sec}) and next that (a) implies
(c) (in Section \ref{minkt.sec}). Finally, (c) trivially implies (b).

Note that a consequence of \ref{main.thm} is that for \emph{every}
polynomial $t(n) \geq (1+\varepsilon) n$, where
      $\varepsilon>0$ is a constant
$t(n)$, mild average-case hardness of $K^{t}$ is
equivalent to the existence of one-way functions.

\section{OWFs from Mild Avg-case $K^t$-Hardness}
In this section, we state our main theorem.
\label{OWF.sec}
\BT \label{owf.thm} Assume there exist polynomials $t(n)>0, p(n)>0$ such that
$K^t$ is $\frac{1}{p(\cdot)}$-HoA. Then there
exists a weak OWF $f$ (and thus also a OWF).
\ET
\begin{proof}
Let $c$ be the constant from Fact \ref{KC.fact}.
Consider the function $f: \{0,1\}^{n+c + \lceil\log (n+c)\rceil} \rightarrow \{0,1\}^*$,
which given an input $\ell || \desc’$ where $|\ell| = \lceil \log(n+c) \rceil$ and $|\desc’|
= n+c$, outputs $\ell || U(\desc,1^{t(n)})$ where $\desc$ is the $\ell$-bit prefix of $\desc'$.
%Yanyi: added some explanation of $M$
This function is only defined over some input lengths, but by an easy
padding trick, it can be transformed into a function $f’$ defined
over all input lengths, such that if $f$ is (weakly) one-way (over the
restricted input lengths), then $f’$ will be (weakly) one-way (over all
input lengths): $f’(x')$ simply truncates its input $x'$ (as little as
possible) so that the (truncated) input $x$ now becomes of length $m= n+c
+ \lceil \log(n+c) \rceil$ for some $n$ and outputs $f(x)$.

We now show if $K^t$ is $\frac{1}{p(\cdot)}$-HoA, then $f$ is a
$\frac{1}{q(\cdot)}$-weak OWF, where
$q(n) = 2^{2c+3}np(n)^2$, which concludes the proof of the theorem.
Assume for contradiction that $f$ is not a $\frac{1}{q(\cdot)}$-weak
  OWF. That is, there exists some $\PPT$ attacker $\A$ that inverts $f$ with probability
% Yanyi: actually it should be 1 - \frac{1}{q(m)}
%$1-\frac{1}{q(n)}$
at least $1-\frac{1}{q(n)} \leq 1 - \frac{1}{q(m)}$
for
infinitely many $m = n+c + \lceil \log(n+c) \rceil$. Fix some such $m,n>2$.
By an averaging argument, except for a fraction $\frac{1}{2p(n)}$ of random tapes $r$
for $\A$, the \emph{deterministic} machine $\A_r$ (i.e., machine $\A$ with
randomness fixed to $r$) fails to invert $f$ with probability at most
$\frac{2p(n)}{q(n)}$. Fix some such ``good'' randomness $r$ for which
$\A_r$ succeeds to invert $f$ with probability $1-\frac{2p(n)}{q(n)}$.

We next show how to use $\A_r$ to compute $K^t$ with high probability over random inputs
$z\in \{0,1\}^n$. Our heuristic $\H_r(z)$ runs $\A_r(i||z)$ for all $i \in [n+c]$ where
$i$ is represented as a $\lceil \log (n+c) \rceil$ bit string, and outputs the length
of the smallest program $\desc$ output by $\A_r$ that produces the string $z$
within $t(n)$ steps.
Let $S$ be the set of strings $z\in \{0,1\}^n$ for which $\H_r(z)$ fails to
compute $K^t(z)$. Note that $\H_r$ thus fails with probability
$$fail_r = \frac{|S|}{2^{n}}.$$ Consider any string $z \in S$ and let
$w=K^t(z)$ be its $K^t$-complexity. By Fact \ref{KC.fact}, we have
that $w \leq n+c$. Since $\H_r(z)$ fails to compute $K^t(z)$, $\A_r$ must
fail to invert $(w||z)$. But, since $w \leq n+c$, the output $(w||z)$ is sampled with
probability
$$\frac{1}{n+c} \cdot \frac{1}{2^{w}}\geq \frac{1}{(n+c)}\frac{1}{2^{n+c}} \geq
  \frac{1}{n 2^{2c+1}} \cdot \frac{1}{2^{n}}$$
%Yanyi: The constant 2 in front of c comes from \frac{1}{(n+c)} \geq \frac{1}{n2^c}, previously it was wrong that claiming $\frac{1}{n+c}\leq \frac{1}{n}$
in the one-way
function experiment, so $\A_r$ must fail with probability at least
$$ |S| \cdot \frac{1}{n 2^{2c+1}} \cdot \frac{1}{2^{n}} = \frac{1}{n
  2^{2c+1}} \cdot \frac{|S|}{2^{n}} =  \frac{fail_r}{n 2^{2c+1}}$$
which by assumption (that $\A_r$ is a good inverter) is at most
that $\frac{2p(n)}{q(n)}$. We thus conclude that $$fail_r \leq \frac{2^{2c+2}np(n)}{q(n)}$$
Finally, by a union bound, we have that $\H$ (using a uniform random
 tape $r$) fails in computing $K^t$ with probability at most
 $$\frac{1}{2p(n)} + \frac{2^{2c+2}np(n)}{q(n)} =
 \frac{1}{2p(n)} + \frac{2^{2c+2}np(n)}{2^{c+3}np(n)^2} = \frac{1}{p(n)}.$$
 Thus, $\H$ computes $K^t$ with probability $1-\frac{1}{p(n)}$ for
 infinitely many $n \in \N$, which contradicts the
 assumption that $K^t$ is $\frac{1}{p(\cdot)}$-HoA.
\end{proof}

\section{Mild Avg-case $K^t$-Hardness from OWFs}
\label{minkt.sec}
We introduce the notion of a (conditionally-secure) \emph{entropy-preserving}
pseudo-random generator (EP-PRG) and next show (1) the existence of a condEP-PRG implies
that $K^t$ is hard-on-average (even to approximate), and (2) OWFs imply condEP-PRGs.

\subsection{Entropy-preserving PRGs}
% yanyi5: removing weak
% We start by defining the notion of a weak Entropy-preserving PRG.
We start by defining the notion of a \emph{conditionally-secure entropy-preserving PRG}.
\BD An efficiently computable function $G: \{0, 1\}^{n} \rightarrow \{0,
1\}^{n+\gamma\log n}$ is a \emph{$\mu(\cdot)$-conditionally secure entropy-preserving pseudorandom
  generator ($\mu$-condEP-PRG)}
if
% yanyi4: add running time bound to g
%raf5: pushed down, don't know what it means for all constant.
%$g$ runs in $(1 + \varepsilon)n$ time for all constant $\varepsilon$, and
there exist a sequence of events $=\{E_n\}_{n\in \N}$ and a
constant $\alpha$ (referred to as the \emph{entropy-loss constant}) such that
the following conditions hold:
\BI
\item {\bf (pseudorandomness):} $\{G(\U_n|E_n)\}_{n \in \N}$ and
  $\{\U_{n+\gamma\log n}\}_{n\in \N}$ are
  $\mu(n)$-indistinguishable;
  \item {\bf (entropy-preserving):} For all sufficiently large $n \in \N$,
    $H(G(\U_n|E_n)) \ge n - \alpha \log n$.
    \EI If for all $n$, $E_n = \{0,1\}^n$ (i.e., there is no
    conditioning), we say that $G$ is an
    \emph{$\mu$-secure entropy-preserving pseudorandom generator ($\mu$-EP-PRG)}.
    \ED
    % raf5
    % yanyi5 adding "for some eps"
    % We say that $G$ has {\em rate-1 efficiency} if its running time on
% inputs of length $n$ is bounded by $n + O(n^{\epsilon})$.
    We say that $G$ has {\em rate-1 efficiency} if its running time on
inputs of length $n$ is bounded by $n + O(n^{\varepsilon})$ for some constant $\varepsilon < 1$.
\subsection{Avg-case $K^t$-Hardness from Cond EP-PRGs}
\BT
Assume that for every $\gamma>1$, there exists a rate-1 efficient $\mu$-condEP-PRG $G: \{0, 1\}^{n} \rightarrow \{0,
1\}^{n+\gamma\log n}$ %with entropy-loss constant $\alpha$,
where
$\mu(n) = 1/n^2$. Then, for every constant $d > 0, \varepsilon > 0$,
for every polynomial $t(n) \geq (1 + \varepsilon)n$, $K^t$ is mildly hard-on-average to
  $(d \log n)$-approximate.
\ET

\begin{proof}
Let $\gamma \geq max(8,8d)$, and let $G': \{0, 1\}^{n} \rightarrow \{0,
%raf5: added rate-1
1\}^{m'(n)}$ where $m'(n) = n+\gamma\log n$ be a rate-1 efficient
$\mu$-condEP-PRG, where $\mu = 1/n^2$.
For any constant $c$, let $G^c(x)$ be a function that computes $G'(x)$ and truncates the last
$c$ bits. It directly follows that $G^c$ is also a rate-1 efficient $\mu$-condEP-PRG (since $G'$ is
so).
%raf2: added t_0.
% yanyi4: replaced t_0 with (1+eps)n
%raf5
%It follows that the running time bound of $g$, $(1+\varepsilon)n$, still bounds the running
%time of $g^c$ for every $c \leq \gamma+1$, let
Consider any $\varepsilon >0$ and any polynomial $t(n) \geq (1+\varepsilon)n$ and let
$p(n)=2n^{2(\alpha+\gamma+1)}$.

Assume for contradiction that there exists some $\PPT$ $\H$ that
  $\beta$-approximates $K^t$ with probability $1-\frac{1}{p(m)}$ for infinitely many
  $m \in \N$, where $\beta(n) = \gamma/8 \log n \geq d \log n$. Since $m'(n+1)-m'(n) \leq \gamma+1$, there must exist some
  constant $c\leq \gamma+1$ such that $\H$ succeeds (to
  $\beta$-approximate $K^t$) with probability
  $1-\frac{1}{p(m)}$ for infinitely many $m$ of the
  form $m= m(n) = n+ \gamma\log n - c$. Let $G(x) = G^c(x)$; recall
  that $G$ is a rate-1 efficient $\mu$-condEP-PRG (trivially, since $G^c$ is so), and let
  $\alpha,\{E_n\}$, respectively, be the entropy loss constant and sequence of
  events, associated with it.

We next show that $\H$ can be used to break the condEP-PRG $G$.
Towards this, recall that a random string has high $K^t$-complexity with
high probability: for $m = m(n)$, we have,
\begin{equation} \label{eq1}
   \Pr_{x \in \{0, 1\}^m}[K^t(x) \ge m - \frac{\gamma}{4}\log n] \ge
   \frac{2^m - 2^{m-\frac{\gamma}{4}\log n}}{2^m} = 1 -
   \frac{1}{n^{\gamma/4}},
   \end{equation}
since the total number of Turing machines with length smaller than
$m-\frac{\gamma}{4}\log n$ is only $2^{m-\frac{\gamma}{4}\log n}$. However, any string
output by the EP-PRG, must have ``low'' $K^t$ complexity: For every
sufficiently large $n,m=m(n)$, we have that,
\begin{equation} \label{eq2}
\Pr_{s\in \{0, 1\}^n}[K^t(G(s)) \ge m- \frac{\gamma}{2}\log n] = 0,
\end{equation}
since $G(s)$ can be represented by combining a seed $s$ of length
$n$ with the code of $G$ (of constant length), and the running time of
$G(s)$ is bounded by $t(|s|) = t(n) \leq t(m)$ for all sufficiently
large $n$, so $K^t(G(s)) =
n+O(1) = (m- \gamma \log n + c) + O(1) \leq m-\gamma/2\log n$ for sufficiently large $n$.

Based on these observations, we now construct a $\PPT$ distinguisher $\A$
breaking $G$. On input $1^n,x$, where $x \in \{0,
1\}^{m(n)}$, $\A(1^n, x)$ lets $w \leftarrow \H(x)$ and outputs 1 if $w \geq m(n) -
\frac{3}{8}\gamma \log n$ and 0 otherwise. Fix some $n$ and $m = m(n)$ for which $\H$ succeeds with probability $\frac{1}{p(m)}$.
The following two claims conclude that $\A$ distinguishes $\U_{m(n)}$ and $G(\U_n \mid E_n)$ with probability at least $\frac{1}{n^2}$.
\begin{clm}
    $\A(1^n,\U_m)$ outputs 1 with probability at least $1 - \frac{2}{n^{\gamma/4}}$.
\end{clm}
\begin{proof}
Note that $\A(1^n,x)$ will output 1 if $x$ is a string with
$K^t$-complexity larger than $m - \gamma/4 \log n$ and $\H$ outputs a
$\gamma/8 \log n$-approximation to $K^t(x)$. Thus,
\begin{align*}
    &\Pr[\A(1^n,x) = 1] \\
    &\ge \Pr[K^t(x) \ge m - \gamma /4 \log n \wedge \H \textrm{ succeeds on }x] \\
    &\ge 1 - \Pr[K^t(x) < m - \gamma /4 \log n] - \Pr[\H \textrm{ fails on }x] \\
    & \ge 1 - \frac{1}{n^{\gamma/4}} - \frac{1}{p(n)}\\
          & \ge 1 - \frac{2}{n^{\gamma/4}}.
\end{align*}
where the probability is over a random $x \leftarrow \U_m$ and the
randomness of $\A$ and $\H$.
\end{proof}

\begin{clm}
    $\A(1^n,G(\U_n \mid E_n))$ outputs 1 with probability at most
    $1-\frac{1}{n}+\frac{2}{n^{\alpha + \gamma}}$
\end{clm}
\begin{proof}
Recall that by assumption, $\H$ fails to $(\gamma/8 \log n)$-approximate $K^t(x)$ for
a random $x \in \{0,1\}^m$ with probability at most $\frac{1}{p(m)}$. By
an averaging argument, for at least a $1 - \frac{1}{n^2}$
fraction of random tapes $r$ for $\H$, the deterministic machine
$\H_r$ fails to approximate $K^t$ with probability at most
$\frac{n^2}{p(m)}$. Fix some ``good'' randomness $r$ such that $\H_r$
%yanyi2: changing computes to approximates
% computes
approximates
$K^t$ with probability at least $1 - \frac{n^2}{p(m)}$.
We next analyze the success probability of $\A_r$. Assume for
contradiction that $A_r$ outputs 1 with probability at least
$1-\frac{1}{n}+\frac{1}{n^{\alpha+\gamma}}$ on input $G(\U_n \mid E_n)$.
Recall that (1) the entropy of $G(\U_n \mid E_n)$ is at least $n - \alpha\log
n$ and (2) the quantity $-\log \Pr[G(\U_n \mid E_n)=y]$ is upper bounded by
$n$ for all $y \in G(\U_n\mid E_n)$ since $H_{\infty}(G(\U_n\mid E_n)) \leq H_{\infty}(\U_n \mid E_n) \leq H_{\infty}(\U_n) = n$.
%Yanyi: added some explanation to the upper bound
By an averaging argument, with probability at least $\frac{1}{n}$, a
random $y \in G(\U_n \mid E_n)$ will satisfy
$$-\log \Pr[G(\U_n \mid E_n) = y] \ge (n - \alpha\log n)  - 1. $$
We refer to an output $y$ satisfying the above condition as being
``good'' and other $y$'s as being ``bad''. Let $S = \{y \in G(\U_n \mid E_n): \A_r(1^n,y) = 1 \wedge y \textrm{ is
  good} \}$, and let $S' = \{y \in G(\U_n \mid E_n): \A_r(1^n,y) = 1 \wedge y \textrm{ is
  bad} \}$. Since
$$\Pr[\A_r(1^n,G(\U_n \mid E_n)) = 1] = \Pr[G(\U_n\mid E_n)\in S
] + \Pr[G(\U_n \mid E_n)\in S'],$$ and $\Pr[G(\U_n \mid E_n)\in S']$ is
at most the probability that $G(\U_n)$ is ``bad'' (which as argued above
is at most $1- \frac{1}{n}$), we have that
%Since $\A_r(g(U_n)\mid E_n)$ outputs 1 with probability at least $1 - \frac{1}{n} +
%\frac{1}{n^{\alpha+\gamma}}$, and $g(U_n)|E_n$ is ``bad'' with
%probability at most $1-\frac{1}{n}$, we have that
$$\Pr[G(\U_n\mid E_n)\in S ]  \ge \left(1 - \frac{1}{n} + \frac{1}{n^{\alpha+\gamma}}\right) - \left(1-\frac{1}{n}\right) = \frac{1}{n^{\alpha+\gamma}}.$$
Furthermore, since for every $y \in S$, $\Pr[G(\U_n\mid E_n)=y] \leq 2^{-n +
  \alpha\log n + 1}$, we also have,
$$\Pr[G(\U_n\mid E_n)\in S ] \leq |S| 2^{-n + \alpha\log n + 1}$$
So,
$$|S| \ge \frac{2^{n -
  \alpha\log n - 1}}{n^{\alpha+\gamma}} = 2^{n-(2\alpha +
\gamma)\log n - 1}$$
However, for any $y \in G(\U_n\mid E_n)$, if $\A_r(1^n,y)$ outputs 1,
then by Equation \ref{eq2}, $\H_r(y) > K^t(y) + \gamma/8$, so $\H$
fails to output a good approximation. (This follows, since by Equation
\ref{eq2}, $K^t(y) < n - \gamma/2 \log n$ and $\A_r(1^n,y)$ outputs 1
only if $\H_r(y) \geq n - \frac{3}{8}\gamma \log n$.)

Thus, the probability that $\H_r$ fails (to output a good
approximation) on a random $y \in \{0, 1\}^m$ is at least
$$|S|/2^m = \frac{2^{n-(2\alpha +
    \gamma)\log n - 1}}{2^{n + \gamma \log n - c}}
\ge 2^{-2 (\alpha+\gamma)\log n - 1} = \frac{1}{2n^{2(\alpha+\gamma)}}$$
which contradicts the fact that $\H_r$ fails with approximate $K^t$ probability at most
$\frac{n^2}{p(m)} < \frac{1}{2n^{2(\alpha+\gamma)}}$ (since $n < m$).

We conclude that for every good randomness $r$, $\A_r$ outputs 1 with
probability at most
$1-\frac{1}{n}+\frac{1}{n^{\alpha+\gamma}}$. Finally, by union bound
(and since a random tape is bad with probability $\leq \frac{1}{n^2}$),
we have that the probability that $\A(G(\U_n \mid E_n))$ outputs 1 is at most
$$\frac{1}{n^2} + \left(1-\frac{1}{n}+ \frac{1}{n^{\alpha+\gamma}}\right) \leq 1-\frac{1}{n} + \frac{2}{n^2},$$
since $\gamma \geq 2$.
\end{proof}
\noindent We conclude, recalling that $\gamma \geq 8$, that $\A$
distinguishes $\U_m$ and $G(\U_n \mid E_n)$ with probability of at
least $$\left(1 - \frac{2}{n^{\gamma/4}}\right) - \left(1-\frac{1}{n}
  + \frac{2}{n^2}\right) \ge \left(1 -\frac{2}{n^2}\right) - \left(1-
\frac{1}{n} + \frac{2}{n^2}\right)  = \frac{1}{n} - \frac{4}{n^2} \geq
\frac{1}{n^2}$$
for infinitely many $n \in \N$.
\end{proof}

\subsection{Cond EP-PRGs from OWFs}
In this section, we show how to construct a condEP-PRG from any OWF.
Towards this, we first recall the construction of \cite{HILL99,Gol01,Yu} of a
PRG from a \emph{regular} one-way function \cite{GKL93}.
\BD \label{regular.def}
    A function $f:\{0, 1\}^* \rightarrow \{0, 1\}^*$ is called
    \emph{regular} if there exists a function $r:\N \rightarrow \N$
%raf1: removing q as it makes things a bit messy
    % and a polynomial $q(\cdot)$
    such that for all sufficiently long $x \in \{0, 1\}^*$,
    %$$\frac{1}{q(|x|)} 2^{r(|x|)} \leq |f^{-1}f(x)| \leq
    %2^{r(|x|)}.$$
    $$2^{r(|x|)-1} \leq |f^{-1}(f(x))| \leq 2^{r(|x|)}.$$
    We refer to $r$ as the regularity of $f$.
\ED
\noindent As mentioned in the introduction, the construction proceeds in the following two steps given a OWF $f$
with regularity $r$.
\BI
%yanyi2: change all f' to \hat{f}
\item  We ``massage'' $f$ into a different OWF $\hat{f}$ having the
  property that there exists some $\ell(n) = n- O(\log n)$ such that $\hat{f}(\U_n)$ is statistically close to $\U_{\ell(n)}$---we will
  refer to such a OWF as being \emph{dense}.
  %yanyi2: we don't require \hat{f} has hardcore bits now
%   ---and additionally, $f'$ has $O(\log n)$ hardcore bits
  This is
  done by applying pairwise-independent hash functions (acting as
  strong extractors) to both the input and the output of
the OWF (parametrized to match the regularity $r$) to
``squeeze'' out randomness from both the input and the output.
$$\hat{f}(s||\sigma_1||\sigma_1)= \sigma_1||\sigma_2|| [h_{\sigma_1}(s)]_{r-O(\log n)} || [h_{\sigma_2} (f(s))]_{n-r-O(\log n)} $$
where $[a]_j$ means $a$ truncated to $j$ bits.
\item We next modify $\hat{f}$ to include additional randomness in
  the input (which is also revealed in the output) to make sure
the function has a hardcore function:
$$f'(s||\sigma_1||\sigma_2||\sigma_{GL}) = \sigma_{GL} || \hat{f}(s||\sigma_1||\sigma_1)$$
  \item We finally use $f'$ to construct a PRG $G^r$ by simply adding the
  the Goldreich-Levin hardcore bits~\cite{GL89}, $GL$, to the output of the function $f'$:
%   $$G^r(s||\sigma_1||\sigma_2||\sigma_{GL})=
%   f'(s||\sigma_1||\sigma_2|| \sigma_{GL}) || GL(s,
%   \sigma_{GL})$$
  $$G^r(s||\sigma_1||\sigma_2||\sigma_{GL})=
  f'(s||\sigma_1||\sigma_2|| \sigma_{GL}) || GL(s || \sigma_1 || \sigma_2,
  \sigma_{GL}))$$
\EI
%  \item By the Goldreich-Levin Theorem \cite{GL}, for every
 %   $\gamma\geq 0$, $f''$ can be modified
  %  into a different OWF $f''$ such that (1) $f''$ is range dense and (2) $f'$ has a $\gamma \log n$
%    hardcore function $GL_{\gamma}$. The final PRG is then $G(s) = f''(s) || GL(s)$.
%    \EI
    We note that the above steps do not actually produce a
    ``fully secure'' PRG as the statistical distance between the output
    of $\hat{f}(\U_n)$ and uniform is only $\frac{1}{\poly(n)}$ as
    opposed to being negligible. \cite{Gol01} thus presents a final
    amplification step to deal with this issue---for our purposes
    it will suffice to get a $\frac{1}{\poly(n)}$ indistinguishability
    gap so we will not be concerned about the amplification step.

    We remark that nothing in the above steps requires $f$ to be a
    one-way function defined on the domain $\{0,1\}^n$---
    % both steps
    all three steps
    still work even for one-way functions defined over domains
    $S$ that are different than $\{0,1\}^n$, as long as a lower bound on the size of the domain
    is efficiently computable (by a minor modification of the
    construction in Step 1 to account for the size of $S$). Let us
    start by formalizing this fact.

\BD Let $\S = \{S_n\}$ be a sequence of sets such that $S_n \subseteq
\{0,1\}^n$ and let $f: S_n \rightarrow \bitset^*$ be a polynomial-time
computable function. $f$ is said to be a \emph{one-way function
  over $\S$ ($\S$-OWF)} if for every $\PPT$
algorithm $\A$, there exists a negligible function $\mu$ such that for
all $n \in \N$,
	$$ \Pr[x \leftarrow S_n; y = f(x) : A(1^n,y) \in f^{-1}(f(x))
        ] \leq \mu(n) $$
        We refer to $f$ as being regular if it satisfies Definition
        \ref{regular.def} with the exception that we only quantify
        over all $n \in N$ and all $x \in S_n$ (as opposed to all $x
        \in \{0,1\}^n)$.
        \ED

\noindent We say that a \emph{family of functions $\{f_i\}_{i \in I}$ is efficiently
computable} if there exists a polynomial-time algorithm $M$ such that
$M(i,x) = f_i(x)$.

\BL [implicit in \cite{Gol01,Yu}]\label{lemma1}
Let $\S = \{S_n\}$ be a sequence of sets such that $S_n \subseteq
\{0,1\}^n$, let $s$ be an efficiently computable function such that
$s(n) \leq \log |S_n|$, and let $f$ be an $\S$-OWF with
regularity $r(\cdot)$.
Then, there exists a constant $c\geq 1$ such that for every $\alpha', \gamma' \geq
0$, there exists an efficiently computable family of functions
$\{f'_i\}_{i \in \N}$, and an efficiently computable function $GL$,
such that the following holds for $\ell(n) = s(n) + 3n^c -2\alpha'
\log n$, $\ell'(n) = \ell(n) + \gamma' \log n$:
\BI
\item {\bf density:} For all sufficiently large $n$, the
  distributions
  \BI
  \item $\left \{x
\leftarrow S_{n}, \sigma_1, \sigma_2, \sigma_{GL}\leftarrow \{0, 1\}^{n^c}:
f'_{r(n)}(x, \sigma_1, \sigma_2, \sigma_{GL}) \right\}$, and
\item $\U_{\ell(n)}$
  \EI
  are $\frac{3}{n^{\alpha'/2}}$-close
in statistical distance.

\item {\bf pseudorandomness:} The ensembles of distributions,
  \BI
  \item $\left \{x
      \leftarrow S_{n}, \sigma_1, \sigma_2, \sigma_{GL}\leftarrow \{0, 1\}^{n^c}:
f'_{r(n)}(x, \sigma_1, \sigma_2, \sigma_{GL}) || GL(x, \sigma_1, \sigma_2, \sigma_{GL})
\right\}_{n \in \N}$, and
  \item
    $\left\{\U_{\ell'(n)}\right\}_{n \in \N}$
    \EI
    are $\frac{4}{n^{\alpha'/2}}$-indistinguishable.
\EI
  \EL
\begin{proof}
  Given a $r(\cdot)$-regular $\S$-OWF $f$, the
  construction of $f'$ has the form

$$f'(s||\sigma_1||\sigma_1 ||\sigma_{GL})=\sigma_{GL}||\sigma_1||\sigma_2||
 [h_{\sigma_1}(s)]_{r-\alpha' \log n} || [h_{\sigma_2} (f(s))]_{s(n)-r-\alpha' \log n} $$
where $|x| =
  n, |\sigma_1|=|\sigma_2| =|\sigma_c| = n^c$, and $GL(x,\sigma_1,\sigma_2,\sigma_{GL})$ is
  simply the Goldreich-Levin hardcore predicate \cite{GL89} outputting $\gamma'
  \log n$ inner products between $x$ and vectors in $\sigma_{GL}$.
The function $f'_r$ thus maps $n' = n + 3n^c$ bits to
  $3n^c + s(n) - 2\alpha' \log n$ bits, and once we add the output of
  $GL$, the total output length becomes $3n^c + s(n) - 2\alpha' \log
  n + \gamma' \log n$ as required.
  The proof in \cite{Gol01,Yu} directly works to show that $\{f_i\},
  GL$ satisfy the requirements stated in the theorem. (For the reader's convenience, we
  present a simple self-contained proof of this in Appendix \ref{regular.sec}.\footnote{This proof may be of independent didactic interest as an
    elementary proof of the existence of PRGs from regular OWFs.})
\end{proof}

We additionally observe that every OWF actually is a regular $\S$-OWFs for
a sufficiently large $\S$.
\BL \label{lemma2}
Let $f$ be an one way function. There exists an integer function $r(\cdot)$ and a sequence of sets $\S = \{S_n\}$ such that $S_n \subseteq \{0, 1\}^n$, $|S_n| \ge \frac{2^n}{n}$, and $f$ is a $\S$-OWF with regularity $r$.
\EL
\begin{proof}
The following simple claim is the crux of the proof:
  \begin{clm}
    For every $n \in \N$, there exists an integer $r_n \in [n]$ such that
    $$\Pr[x \leftarrow \{0, 1\}^n:  2^{r_n-1} \leq |f^{-1}( f(x)|) \leq 2^{r_n}] \ge \frac{1}{n}.$$
\end{clm}
\begin{proof}
    For all $i \in [n]$, let
    $$w(i) = \Pr[x \leftarrow \{0, 1\}^n: 2^{i-1} \leq |f^{-1}( f(x))| \leq 2^i].$$
    Since for all $x$, the number of pre-images that map to $f(x)$ must be in the range of $[1, 2^n]$, we know that $\sum_{i=1}^n w(i)=1$. By an averaging argument, there must exists such $r_n$ that $w(r_n) \ge \frac{1}{n}$.
\end{proof}

    Let $r(n) = r_n$ for every $n \in N$, $S_n = \{x \in \{0,
    1\}^n:2^{r(n)-1} \leq |f^{-1}( f(x))| \leq 2^{r(n)}]\}$; regularity
    of $f$ when the input domain is restricted to $\S$ follows
    directly. It only remains to show that $f$ is a $\S$-OWF; this
    follows directly from the fact that the set $S_n$ are dense in
    $\{0,1\}$. More formally, assume for contradiction that there
    exists a $\PPT$ algorithm $\A$ that inverts $f$ with probability
    $\varepsilon(n)$ when the input is sampled in $S_n$. Since $|S_n|
    \ge \frac{2^n}{n}$, it follows that $\A$ can invert $f$ with
    probability at least $\varepsilon(n)/n$ over uniform distribution,
    which is a contradiction (as $f$ is a OWF).
\end{proof}

By combining Lemma \ref{lemma1} and Lemma \ref{lemma2}, we can
directly get an EP-PRG defined over a subset $\S$. We next turn to
% yanyi5: replacing 'weak' by \mu-con
showing how to instead get a \emph{$\mu$-conditionally secure} EP-PRG that is defined over $\{0,1\}^n$.
%raf5: adding rate 1
% yanyi5: removing rate 1 and adding \mu-cond
\BT \label{thm:allmu}
% Assume that one way functions exist. Then, for every $\gamma > 1, \mu > 1$,
% there exists a $\mu$-condEP-PRG $g:\{0, 1\}^{n'} \rightarrow \{0, 1\}^{n' + \gamma \log {n'}}$.
% yanyi5: there should be a polynomial that bounds the running time of all g_{\mu} for every \mu. Otherwise the padding trick doesn't work
Assume that one way functions exist. Then, there exists a polynomial $t_0(\cdot)$ such that for every $\gamma > 1, \delta > 1$,
% yanyi5: replacing 'g' by 'G''
% yanyi5: \mu = n^{-\delta}
there exists a $\left(\frac{1}{{n}^\delta}\right)$-condEP-PRG $G'_{\delta, \gamma}:\{0, 1\}^{n} \rightarrow \{0, 1\}^{n + \gamma \log {n}}$ with running time bounded by $(\gamma + \delta) t_0(n)$.
\ET

\begin{proof}
    By Lemma~\ref{lemma2}, there exists a
    sequence of sets $\S= \{S_n\}$ such that $S_n \subseteq \{0,1\}^n,
    |S_n| \ge \frac{2^n}{n}$, a function $r(\cdot)$, and an $\S$-OWF $f$ with regularity
    $r(\cdot)$.
    Let $s(n) = n - \log
    n$ (to ensure that $s(n) \leq \log |S_n|$). Let $c$ be the
    constant guaranteed to exist by Lemma~\ref{lemma1} w.r.t. $\S$ and $f$.
    Consider any $\delta,\gamma>1$ and define $\alpha' = 8c\delta$ and $\gamma' = (c+1)\gamma +
    2\alpha' + 3$, and define $\ell(n),\ell'(n)$ just as in the
    stament of Lemma~\ref{lemma1}, namely, $\ell(n) = s(n) +
    3n^c -2 \alpha' \log n$ and $\ell'(n) = \ell(n) + \gamma' \log n$.
    Let $\{f'_i\}_{i \in \N}$ and $GL$ be the functions guaranteed to
    exists by Lemma~\ref{lemma1} w.r.t. $\alpha',\gamma'$,
%
%, there exists a constant $c$ such that for
%    every $\alpha', \gamma' \geq 0$, there exists an efficiently
 %   computable family of functions $\{f'_i\}_{i \in \N}$, and an
  %  efficiently computable function $GL$ satisfying the \emph{density} and
   % \emph{pseudorandomness} properties described in Lemma
   % \ref{lemma1}. %Consider some
    % yanyi5: adding \mu <=> 1/n^p
    % $\alpha' \geq 8c$
    % yanyi6: making \alpha' and \gamma' explicitly here
    % $\alpha' \geq 8c\delta$
    %  and any $\gamma' \geq 0$.
    %$\alpha' = 8c\delta$ and $\gamma' = (c+1)\gamma +
    %2\alpha' + 3.$
    %Let $\ell(n) = s(n) +
    %3n^c -2 \alpha' \log n$, $\ell'(n) = \ell(n) + \gamma' \log n$
and consider the function $G_{\delta, \gamma}:\{0, 1\}^{\log n + n + 3n^c}
    \rightarrow \{0, 1\}^{\ell'(n)}$ defined as follows:
$$G_{\delta, \gamma}(i, x, \sigma_1,\sigma_2, \sigma_{GL}) = f'_{i}(x,
    \sigma_1,\sigma_2,\sigma_{GL}) || GL(x, \sigma_1,\sigma_2, \sigma_{GL})$$
    where $|i| = \log n, i\in [n], |x| = n,
    |\sigma_1|=|\sigma_2|=|\sigma_{GL}| = n^c$.
    Let $n' = n'(n) =\log
    n + n + 3n^c$ denote the input length of $G_{\delta, \gamma}$. Let $\{E_{n'(n)}\}$ be
    a sequence of events where
    $$E_{n'(n)} = \{i, x,
    \sigma_1,\sigma_2,\sigma_{GL} : i = r(n), x \in S_{n},
    \sigma_1,\sigma_2,\sigma_{GL} \in \{0,1\}^{n^c}\} $$

\noindent    Note that the two distributions,
    \BI
    \item $\{x
    \leftarrow S_{n}, \sigma_1,\sigma_2,\sigma_{GL} \leftarrow \{0, 1\}^{n^c}: f'_{r(n)}(x, \sigma_1,\sigma_2,\sigma_{GL}) ||
    GL(x,\sigma_1,\sigma_2,\sigma_{GL})\}_{n \in \N}$, and

  \item $G_{\delta, \gamma}(\U_{n'} \mid E_{n'})$
    \EI
     are identically distributed. It follows
    from Lemma~\ref{lemma1} that $\{G_{\delta, \gamma}(\U_{n'} \mid E_{n'})\}_{n \in \N}$
    and $\{\U_{\ell'(n)}\}_{n\in \N}$ are
    $\frac{4}{n^{\alpha'/2}}$-indistinguishable. Note that for
    % yanyi5: adding \mu <=> 1/n^\delta
    % $\alpha' \geq 8c$
    $\alpha' = 8c\delta$,
    we have that $\frac{4}{n^{\alpha'/2}} =  \frac{4}{n^{4c\delta}}
    \leq \frac{1}{n'(n)^\delta}$ for sufficiently large $n$. Thus, $G_{\delta, \gamma}$ satisfies the
          pseudorandomness property of a $\left(\frac{1}{n'^\delta}\right)$-cond EP-PRG.

    We further show that the output of $G_{\delta, \gamma}$ preserves entropy. Let
    $X_n$ be a random variable uniformly distributed over $S_n$. By
    Lemma~\ref{lemma1}, $f'_{r(n)}(X_n, \U_{3n^c})$ is
    $\frac{4}{n^{\alpha'/2}} \leq  \frac{4}{n^{4c\delta}}
          \leq \frac{1}{\ell(n)^2}$
      close to $\U_{\ell(n)}$ in statistical distance for sufficiently
      large $n$. By Lemma~\ref{lemma:SDtoH} it thus holds that $$H(f'_{r(n)}(X_n, \U_{3n^c})) \ge \ell(n) - 2.$$
    It follows that
    $$H(f'_{r(n)}(X_n, \U_{3n^c}), GL(X_n, \U_{3n^c})) \ge H(f'_{r(n)}(X_n, \U_{3n^c})) \ge \ell(n) - 2.$$
    Notice that $G_{\delta, \gamma}(\U_{n'} \mid E_{n'})$ and $(f'_{r(n)}(X_n,
    \U_{3n^c}), GL(X_n, \U_{3n^c}))$ are identically distributed, so on inputs of length $n'=n'(n)$, the entropy loss of $G_{\delta, \gamma}$ is $n' - (\ell(n) - 2) \leq (2\alpha'+2)\log n + 2 \leq (2\alpha'+4)\log n'$, thus $G_{\delta, \gamma}$ satisfies the entropy-preserving property (by setting the entropy loss $\alpha$ in cond EP-PRG to be $(2\alpha'+4)$).

    The function $G$ maps $n' = \log n + n + 3n^c$ bits to $\ell'(n)$
    bits, and it is thus at least
    %$(\gamma' - \alpha' - 2) \log n \leq \ell'(n) - n'$
    %yanyi: at least <=> \ge, adjusting the order such that it's easy to read
    $\ell'(n) - n' \ge (\gamma' - 2\alpha' - 2) \log n$
    -bit expanding. Since $n' \leq n^{c+1}$ for
    sufficiently large $n$ and recall that $\gamma' = (c+1)\gamma +
    2\alpha' + 2$, $G_{\delta, \gamma}$ will expand its input by at least
    % yany: the same thing here
    $(\gamma' - 2\alpha' - 2) \log n \ge (c+1)\gamma \log n \ge \gamma \log n'$
    bits.

    Notice that although $G_{\delta, \gamma}$ is only defined over some input
    lengths $n = n'(n)$, by taking ``extra'' bits in the input and appending them
    to the output, $G_{\delta, \gamma}$ can be transformed to a cond EP-PRG $G'_{\delta, \gamma}$
    defined over all input lengths: $G'_{\delta, \gamma}(x')$ finds a prefix $x$ of
    $x'$ as long as possible such that $|x|$ is of the form $n' = \log n + n
    + 3n^c$ for some $n$, rewrites $x'=x||y$, and outputs
    $G_{\delta, \gamma}(x)||y$. The entropy preserving and the pseudorandomness property of
    $G'_{\delta, \gamma}$ follows directly; finally, note that if $|x'|$ is
    sufficiently large, it holds that $n^{c+1} \ge |x'|$, and thus by
    the same argument as above, $G'_{\delta, \gamma}$ will also expand its input by at least $\gamma \log |x'|$ bits.

    % yanyi5: adding running time analysis
    We finally show that there exists some polynomial $t_0(n')$ such
    that for every $\delta, \gamma > 1$,
    $(\gamma + \delta) t_0(n')$ bounds the running time of
    $G'_{\delta,\gamma}$ on
    inputs of length $n'$. To see this, note that the OWF used in this construction can be
    assumed to have some fixed polynomial running time.
    The hash function and the GL hardcore function take (no more than) $O(n^c)$ time to output one bit, and in total the hash function outputs at most $O(n)$ bits,
    so the running time of the hash function is $O(n^{c+1})$.
    (If $\delta$ increases, then $\alpha'$ increases---recall that $\alpha' \ge 8c\delta$---and the hash function outputs fewer bits and runs faster.)
    On the other hand, for all $\gamma, \delta$, $G$ outputs $\gamma'\log(n) = ((c+1)\gamma + 2\alpha' + 2)\log n = (\gamma + \delta) O(\log n)$ GL hardcore bits.
    Thus, for any $\gamma, \delta$, $G'$ runs in $\poly(n) + O(n^{c+1}) + (\gamma + \delta) O(n^c\log n) \leq (\gamma + \delta)t_0(n')$ time for some polynomial $t_0(n')$ over input of length $n'$.
  \end{proof}

%  yanyi5: a separate padding trick here
We now use a standard padding trick to obtain a rate-1 efficient
$\mu$-cond EP-PRG: we simply output the first $n-\ell$ bits unchanged,
and next apply a cond EP-PRG on the last $\ell$ bits. Since we only
have a cond EP-PRG that satisfies inverse polynomial (as opposed to
negligible) indistinguishability, we need to be a bit careful with the
choice of the parameters.

\BT
Assume that one way functions exist. Then, for every $\gamma>1$, there exists a rate-1 efficient $\mu$-cond EP-PRG $G_{\gamma}: \{0, 1\}^{n} \rightarrow \{0,
1\}^{n+\gamma\log n}$, where $\mu(n) = 1/n^2$.
\ET

\begin{proof}
    Let $t_0(\cdot)$ be the polynomial guaranteed to exist due to
    Theorem~\ref{thm:allmu}. Let $c_0$ be a constant such that
    $O(n^{c_0}) \ge t_0(n)$. Consider any $\gamma>1$, and let $\gamma'
    = 2c_0\gamma$ and $\delta'
    = 4c_0$ and $\mu'(n) = \frac{1}{n^{\delta'}}$. By
    Theorem~\ref{thm:allmu}, there exists a $\mu'$-cond EP-PRG
    $G'_{\delta',\gamma'} : \bitset^{n} \rightarrow \bitset^{n + \gamma'
      \log n}$; let $\alpha'$ its associated entropy-loss constant.
    Consider a function $G_{\gamma}:\bitset^n \rightarrow \bitset^{n +
      \gamma \log n}$ defined as follows:
    $$G_{\gamma}(s_0||s_1) = s_0||G'_{\delta',\gamma'}(s_1)$$
    where $|s_{1}| = n^{\frac{1}{2c_0}}$.
    Note that $|G'_{\delta',\gamma'}(s_1)| = |s_{1}| + \gamma' \log
    |s_{1}| = n^{\frac{1}{2c_0}} + \gamma' \log (n^{\frac{1}{2c_0}}) =
    n^{-2c_0} + \gamma \log n$, so $G_{\gamma}$ is $(\gamma\log n)$-bit expanding.
    Furthermore, the entropy-loss of $G_{\gamma}$ is $\alpha' \log
    (n^{\frac{1}{2c_0}}) = \alpha \log n$ for some constant $\alpha =
    \frac{\alpha'}{2c_0}$.
    Since the running time of $G'_{\delta',\gamma'}$ is bounded by $(\gamma' + \delta') t_0(n^{\frac{1}{2c_0}}) \leq O(n^{\frac{1}{2}})$, the running time of $G_{\gamma}$ is $|s_{0}| + O(n^{\frac{1}{2}}) \leq n + O(n^{\frac{1}{2}})$.
    Finally, it holds that $\mu'(|s_{1}|) = \mu'(n^{\frac{1}{2c_0}}) =
    \frac{1}{n^2}$, so we conclude that $G_\gamma$ is a rate-1 efficient
    $\mu$-cond EP-PRG for $\mu(n) = \frac{1}{n^2}$, that expand $n$
    bits to $(n + \gamma \log n)$ bits.
\end{proof}

\section{Acknowledgements}
We are very grateful to Eric Allender, Kai-min Chung, Naomi Ephraim, Cody Freitag, Johan H\aa stad, Yuval
Ishai, Ilan Komargodski, Rahul Santhanam, and abhi shelat for extremely helpful comments. We are
also very grateful to the anonymous FOCS reviewers.
%\newpage
\bibliography{crypto}
\bibliographystyle{alpha}
\appendix
\section{Proof of Lemma~\ref{lemma1}}
\label{regular.sec}
In this section we provide a proof of Lemma
\ref{lemma1}. As mentionned in the main body, the proof of this lemma
readily follows using the proofs in \cite{HILL99,Gol01,Yu}, but for the
convenience of the reader, we provide a simple self-contained proof of
the lemma (which may be useful for didactic purposes).
We start by recalling the Leftover Hash Lemma \cite{HILL99} and the
Goldreich-Levin Theorem \cite{GL89}.

\paragraph{The Leftover Hash Lemma}
We recall the notion of a universal hash function \cite{car79}.
\BD
Let ${\cal H}_m^n$ be a family of functions where $m<n$ and each
function $h \in {\cal H}^n_m$ maps $\bitset^n$ to $\bitset^m$. We say that ${\cal H}^n_m$ is a \emph{universal hash family} if (i) the functions $h_{\sigma} \in {\cal H}_m^n$ can be described by a string $\sigma$ of $n^c$ bits where $c$ is a universal constant that does not depend on $n$; (ii) for all $x \neq x' \in \bitset^n$, and for all $y,y' \in \bitset^m$ $$ \Pr[h_\sigma \leftarrow \H^n_m:h_\sigma(x)=y \; \mbox{ and } \; h_\sigma(x')=y'] = 2^{-2m} $$
\ED

\noindent It is well-known that truncation preserves pairwise independence; for completeness, we recall
the proof:
\BL \label{lem:truncate}
If $\H_m^n$ is a universal hash family and $\ell \leq n$, then $\H'^n_{\ell} = \{h_\sigma \in \H^n_m: [h_\sigma]_{\ell}\}$ is also a universal hash family.
\EL
\begin{proof}
    For every $x \neq x' \in \bitset^n, y, y'\in \bitset^\ell$,
    \begin{align*}
        &\Pr[h_\sigma \leftarrow \H^n_m; [h_\sigma(x)]_{\ell} = y\; \mbox{ and } \; [h_\sigma(x')]_{\ell}=y'] \\
        &=\sum_{z\in \bitset^n, [z]_\ell = y}\sum_{z'\in \bitset^n, [z']_\ell = y'} Pr[h_\sigma \leftarrow \H^n_m; h_\sigma(x) = z\; \mbox{ and } \; h_\sigma(x')=z'] \\
        &=2^{-2\ell}.
    \end{align*}
\end{proof}
\noindent Carter and Wegman demonstrate the existence of efficiently computable
universal hash function families.
\BL [\cite{car79}] \label{lem:hashconstruct}
There exists a polynomial-time computable function $H:\bitset^n \times
\bitset^{n^c} \rightarrow \bitset^n$ such that for every $n$, $\H^n_n =
\{h_{\sigma}:\sigma \in \bitset^{n^c}\}$ is a universal hash family,
where $h_\sigma:\{0, 1\}^n \rightarrow \bitset^n$ is defined as $h_\sigma(x) = H(x, \sigma)$.
\EL

\noindent We finally recall the Leftover Hash Lemma.
\BL [Leftover Hash Lemma (LHL) \cite{HILL99}] \label{lemma:LHL}
For any integers $d < k \leq n$, let $\H^n_{k-d}$ be a universal hash family
where each $h \in \H^n_{k-d}$ maps $\bitset^n$ to $\bitset^{k-d}$.
Then, for any random variable $X$ over $\bitset^n$ such that $H_{\infty}(X) \ge k$, it holds that
$$\SD( (H_{k-d}^n , H_{k-d}^n(X)), (H_{k-d}^n , \U_{k-d})) \leq 2^{-\frac{d}{2}},$$
where $H_{k-d}^n$ denotes a random variable uniformly distributed over $\H_{k-d}^n$.
\EL

\paragraph{Hardcore functions and the Goldreich-Levin Theorem}
We recall the notion of a hardcore function and the Goldreich-Levin
Theorem \cite{GL89}.

\BD \label{def:hardcore}
A function $g:\{0, 1\}^n \rightarrow \bitset^{v(n)}$ is called a
\emph{hardcore function for $f:\bitset^n \rightarrow \bitset^*$ over
  $\S = \{S_n \subseteq \bitset^n \}_{n\in N}$} if the following
ensembles are indistinguishable:
\BI
\item $\{x \leftarrow S_n: f(x) || g(x) \}_{n \in \N}$
    \item $\{x \leftarrow S_n: f(x) || \U_{v(n)} \}_{n \in \N}$
\EI
%for every $\PPT$ $\A$ there
%exists a negligible function $\mu$ such that for all $n \in \N$,
%$$|\Pr[x \leftarrow S_n; \A(1^n, f(x), g(x))] - \Pr[x \leftarrow S_n; \A(1^n, f(x), \U_{v(n)})]| \leq \mu(n).$$
\ED
\noindent While the Goldreich-Levin theorem is typically stated for one-way
functions $f$, it actually applies to any randomized function
$f(x,\U_m)$ of $x$ that \emph{hides} $x$. Note that hiding is a
weaker property than one-wayness (where the attacker is only required
to find \emph{any} pre-image, and not necessarily the pre-image $x$ we computed the
function on).
Such a version of the Goldreich-Levin theorem was explicitly stated in
e.g., ~\cite{HHR06} (using somewhat different terminology).

\BD\label{def:hiding}
A function $f: \{0,1\}^n \times \{0,1\}^{m(n)} \rightarrow \{0,1\}^*$ is said to be \emph{entropically-hiding
  over $\S = \{S_n\}_{n \in \N}$ ($\S$-hiding}) if for every $\PPT$
algorithm $\A$, there exists a negligible function $\mu$ such that for
all $n \in \N$,
	$$ \Pr[x \leftarrow S_n, r \leftarrow \bitset^{m(n)}; A(1^n,f(x,r)) = x ] \leq \mu(n) $$
\ED

\BT [\cite{GL89}, also see Theorem 2.12 in~\cite{HHR06}] \label{thm:GL}
There exists some $c$ such that for every $\gamma$, and every $m(\cdot)$, there exists a polynomial-time computable function $GL:\{0,
1\}^{n+m(n) + n^c} \rightarrow \bitset^{\gamma \log n}$ such that the following holds:
Let $\S = \{S_n \subseteq \bitset^n \}_{n\in N}$ and let $f: \{0,1\}^n \times \{0,1\}^{m(n)} \rightarrow \{0,1\}^*$ be
$\S$-hiding. Then $GL$ is a hardcore function for $f': \{0,1\}^n
\times \{0,1\}^{m(n)} \times \{0,1\}^{n^c} \rightarrow \bitset^{*}$,
defined as $f'(x,r,\sigma) = \sigma || f(x,r)$.
\ET

\noindent Given these preliminaries, we are ready to present the proof of Lemma ~\ref{lemma1}.
\paragraph{Proof of Lemma ~\ref{lemma1}}
Let $\S = \{S_n\}$ be a sequence of sets such that $S_n \subseteq
\{0,1\}^n$, let $s$ be an efficiently computable function such that
$s(n) \leq \log |S_n|$, and let $f:S_n \rightarrow \{0, 1\}^n$ be a $\S$-OWF with
regularity $r(n)$.
By Lemma \ref{lem:hashconstruct} and Lemma \ref{lem:truncate}, there exists some constant $c$ and a polynomial-time computable function $H:\bitset^n \times
\bitset^{n^c} \rightarrow \bitset^n$ such that for every $n, m\geq n$, $\H_m^n =
\{h'_{\sigma}:\sigma \in \bitset^{n^c}\}$ is a universal hash family,
where $h'_\sigma = [h_{\sigma}]_m $ and $h_\sigma(x) = H(x, \sigma)$.
We consider a ``massaged'' function $f_i$, obtained by hashing the
input and the output of $f$: $f_i: S_n \times \bitset^{n^c} \times \bitset^{n^c} \rightarrow \bitset^{2n^c}\times \bitset^{i - \alpha' \log
  n} \times \bitset^{s(n)-i-\alpha' \log n}$
$$f_i(x,\sigma_1,\sigma_2) =
\sigma_1||\sigma_2||[h_{\sigma_1}(x)]_{i - \alpha' \log
  n}||[h_{\sigma_2}(f(x))]_{s(n)-i-\alpha' \log n}$$
where $n = |x|$
and show that the function $\hat{f} (x,(\sigma_1,\sigma_2)) = f_{r(n)}(x,\sigma_1,\sigma_2)$ is $\S$-hiding.
\BCM The function $\hat{f}(\cdot,\cdot)$ is $\S$-hiding.
\ECM
\begin{proof}
Assume for contradiction that there exists a $\PPT$ $A$ and a polynomial $p(\cdot)$
such that for infinitely many $n \in \N$,
    $$\Pr[x \leftarrow S_n, \sigma_1,\sigma_2 \leftarrow
    \bitset^{n^c} :  \A(1^n, f_{r(n)}(x, \sigma_1, \sigma_2)) = x] \ge
    \frac{1}{p(n)}$$
      That is,
     $$\Pr[x \leftarrow S_n, \sigma_1,\sigma_2 \leftarrow
    \bitset^{n^c}: \A(1^n, \sigma_1||\sigma_2||[h_{\sigma}(x)]_{r(n) - \alpha' \log
  n}||[h_{\sigma_2}(f(x))]_{s(n)-r(n)-\alpha' \log n}) = x] \ge
    \frac{1}{p(n)}.$$
       We show how to use $\A$ to invert $f$. Consider the $\PPT$ $\A'(1^n,y)$ that samples
       $\sigma_1,\sigma_2 \leftarrow \{0,1\}^{n^c}$ and a ``guess'' $z \leftarrow \{0,1\}^{r(n) - \alpha' \log
  n}$, and outputs
$\A(1^n,\sigma_1||\sigma_2||z||[h_{\sigma_2}(y)]_{s(n)-r(n)-\alpha' \log
  n})$.
Since the guess is
correct with probability $2^{-r(n) + \alpha' \log n} \geq 2^{-r(n)}$,
we have that
     $$\Pr[x \leftarrow S_n : \A'(1^n, f(x)) = x] \ge \frac{2^{-r(n)}}{p(n)}.$$
Since the any $y \in f(S_n)$ has at least $2^{r(n)-1}$ pre-images (since $f$ is $r(n)$-regular over $\S$), we have that
$$\Pr[x \leftarrow S_n : \A'(1^n, f(x)) = x] \ge \Pr[x \leftarrow S_n :
\A'(1^n, f(x)) \in f^{-1}(f(x))] \times 2^{-r(n)+1}.$$
Thus,
$$\Pr[x \leftarrow S_n : \A'(1^n, f(x)) \in f^{-1}(f(x))] \ge 2^{-r(n)+1}
\times \Pr[x \leftarrow S_n : \A'(1^n, f(x)) = x] \geq
\frac{1}{2p(n)}$$
which contradicts that $f$ is an $\S$-OWF.
\end{proof}

Next, consider $f'_i (s,\sigma_1,\sigma_2,\sigma_{GL}) =
\sigma_{GL}||f_i(s,\sigma_1,\sigma_2)$, and the hardcore function $GL$
guaranteed to exists by Theorem \ref{thm:GL}. Since $\hat{f}$ is $\S$-hiding, by Theorem \ref{thm:GL}, the following ensembles are
indistinguishable:
\BI
\item $\{x \leftarrow S_n, \sigma_1, \sigma_2,
  \sigma_{GL} \leftarrow \{0,1\}^{n^c}:
  f'_{r(n)}(x,\sigma_1,\sigma_2,\sigma_{GL}) ||
  GL(x,(\sigma_1,\sigma_2),\sigma_{GL}) \}_{n \in \N}$

\item $\{x \leftarrow S_n, \sigma_1, \sigma_2,
  \sigma_{GL} \leftarrow \{0,1\}^{n^c}:
  f'_{r(n)}(x,\sigma_1,\sigma_2,\sigma_{GL}) || \U_{\gamma' \log n}
  \}_{n \in \N}$
  \EI
  We finally show that $\{x \leftarrow S_n, \sigma_1, \sigma_2,
  \sigma_{GL} \leftarrow \{0,1\}^{n^c}:
  f'_{r(n)}(x,\sigma_1,\sigma_2,\sigma_{GL}) \}$ is
  $\frac{3}{n^{\alpha'/2}}$ close to uniform for every $n$, which will conclude the
    proof of both the pseudorandomness and the density properties by a
    hybrid argument. Let
    $X$ be a random variable uniformly distributed over $\S_n$, and
    let $R_1,R_2$,$R_{GL}$ be random variables uniformly distributed over
    $\{0,1\}^{n^c}$.
    Let
    $$\mathtt{REAL} = f'_{r(n)}(X,R_1,R_2,R_{GL}) = R_{GL}||R_1||R_2|| [h_{R_1}(X)]_{r(n) - \alpha' \log n}, [h_{R_2}(f(X))]_{s(n)-r(n) - \alpha' \log
      n}$$
    We observe:
    \BI
  \item For every $y \in f(S_n)$, $H_{\infty}(X | f(X) = y) \geq r(n)-1$ due to the fact that $f$
    is $r(n)$-regular; by the LHL (i.e., Lemma \ref{lemma:LHL}), it
    follows that $\mathtt{REAL}$ is $\frac{2}{n^{\alpha'/2}}$ close in
      statistical distance to
      $$\mathtt{HYB}_1 = R_{GL}||R_1||R_2|| \U_{r(n) - \alpha' \log n}|| [h_{R_2}(f(X))]_{s(n)-r(n) - \alpha' \log
      n}$$

\item $H_{\infty}(f(X)) \geq s(n)- r(n)$ due to the fact that $f$
    is $r(n)$-regular and $|S_n| \ge s(n)$; by the LHL, it
    follows that $\mathtt{HYB}_1$ is
    $\frac{1}{n^{\alpha'/2}}$ close in
      statistical distance to
      $$\mathtt{HYB}_2 = R_{GL}||R_1||R_2|| \U_{r(n) - \alpha' \log n}|| \U_{s(n)-r(n) - \alpha' \log
      n} = \U_{s(n)+3n^c-2\alpha' \log n}$$
    \EI
Thus, $\mathtt{REAL}$ is $\frac{3}{n^{\alpha'/2}}$-close to uniform,
    which concludes the proof.
\end{document}